%% file: mainUCondense_article.tex
\documentclass[11pt]{article}
\usepackage{authblk}
\usepackage[utf8]{inputenc}
\usepackage[T2A]{fontenc}
\usepackage[english]{babel}
\usepackage{fullpage}

\usepackage{amsfonts,amssymb,amsmath, amsthm}

\usepackage{enumerate}

\newtheorem{claim}{Claim}
\newtheorem{lemma}{Lemma}
\newtheorem{theorem}{Theorem}

\newtheorem{corollary}{Corollary}
\newtheorem{definition}{Definition}
\newtheorem{remark}{Remark}

\newtheorem{proposition}{Proposition}

\usepackage{xspace}
\usepackage{graphicx}
\usepackage{paralist}

\usepackage{comment}

\usepackage{cite}

\usepackage{url}
\usepackage{hyperref}
\hypersetup{
  colorlinks=true,
  linkcolor=blue,
  citecolor=green
}

\usepackage[ruled,commentsnumbered,linesnumbered]{algorithm2e}
\usepackage[shadow,colorinlistoftodos,textwidth=22mm]{todonotes}
	\usepackage{pstricks,pst-node}

\author[1]{Maciej Sk\'{o}rski\thanks{{\texttt maciej.skorski@gmail.com}. Research supported by the WELCOME/2010-4/2 grant.}}
\author[2]{Alexander~Golovnev\thanks{\texttt alexgolovnev@gmail.com}}
\author[3]{Krzysztof Pietrzak\thanks{{\texttt pietrzak@ist.ac.at}. Research supported by ERC starting grant (259668-PSPC).}}
\affil[1]{University of Warsaw}
\affil[2]{New York University}
\affil[3]{IST Austria}

\newtheorem{Claim}{Claim}

\renewcommand{\ge}{\geqslant}
\renewcommand{\le}{\leqslant}

\newcommand{\HAvg}[1]{\widetilde{\mathbf{H}}_{\infty}\left(#1\right)}

\newcommand{\nstrings }{\{0,1\}^{n}}
\newcommand{\mstrings }{\{0,1\}^{m}}

\newcommand{\cD}{\mathsf{ D } }
\newcommand{\E}{\mathbb{E}}

\SetProcNameSty{textsc}

\newcommand{\eps}{\epsilon}
\newcommand{\Ex}{Ex}

\newcommand{\algA}{{\sf A}}
\newcommand{\algB}{{\sf B}}
\newcommand{\algC}{{\sf C}}
\newcommand{\algD}{{\sf D}}

\newcommand{\algP}{{\sf P}}

\newcommand{\minentr}{H_{\infty}}
\newcommand{\hillentr}[1]{H_{#1}^{\sf{HILL}}}

\newcommand{\metricentrdb}[1]{H_{#1}^{{\sf Metric},det,\{0,1\}}}
\newcommand{\metricentrdr}[1]{H_{#1}^{{\sf Metric},det,[0,1]}}
\newcommand{\metricentrrb}[1]{H_{#1}^{{\sf Metric},rand,\{0,1\}}}

\newcommand{\metricentrc}[1]{H_{#1}^{{\sf Metric},class,range}}
\newcommand{\avminentr}{\widetilde{H}_{\infty}}
\newcommand{\expt}[1]{\underset{#1}{\mathbb{E}}}

\newcommand{\Uentr}[1]{H_{#1}^{\sf{unp}}}
\newcommand{\UentrX}[1]{H_{#1}^{\sf{*unp}}}

\newcommand{\supp}[1]{\ensuremath{\operatorname{supp}[#1]}}
\newcommand{\epsclose}[1]{\sim_{#1}}

\newcommand{\bin}{\{0,1\}}
\newcommand{\ext}{{\sf Ext}}

\newcommand{\GL}{{\sf GL}}
\newcommand{\cond}{{\sf Cond}}

\newcommand{\F}{{\sf F}}

\newcommand{\Eq}{\mathsf{Eq}}

\title{Condensed Unpredictability}

\date{}

\begin{document}

\maketitle

\begin{abstract}
We consider the task of deriving a key with high HILL
entropy  (i.e., being computationally indistinguishable from 
a key with high min-entropy) from an unpredictable source.

Previous to this work, the only known way to transform unpredictability into 
a key that was $\eps$ indistinguishable from having min-entropy was via 
pseudorandomness, for example by Goldreich-Levin (GL) hardcore bits.
This approach has the inherent limitation that from a source with $k$ bits of unpredictability entropy one can derive a key of length (and thus HILL entropy) 
at most $k-2\log(1/\epsilon)$ bits. In many settings, e.g. when dealing with biometric data, such a $2\log(1/\epsilon)$ bit entropy loss in not an option.

Our main technical contribution is a theorem that states that in the high entropy regime, unpredictability implies HILL entropy. 
Concretely,  any variable $K$ with $|K|-d$ bits of unpredictability entropy has the same amount of so called 
metric entropy (against real-valued, deterministic distinguishers), which is known to imply the same amount of HILL entropy. 
The loss in circuit size in this argument is exponential in the entropy gap $d$, and thus this result only applies for small $d$ (i.e., where the
size of distinguishers considered is exponential in $d$).

To overcome the above restriction, we investigate if it's possible to first ``condense'' unpredictability entropy and make the entropy gap small. We show that any source with 
$k$ bits of unpredictability can be condensed into a source of length $k$ with $k-3$ bits of unpredictability entropy. 
Our condenser simply ``abuses" the GL construction and derives a $k$ bit key from a source with $k$ bits of unpredicatibily. The original GL theorem 
implies nothing when extracting 
that many bits, but we show that in this regime, GL still behaves like a ``condenser" for unpredictability.
This result comes with two caveats  (1) the loss in circuit size is exponential in $k$ and (2) we require that the source we start with has \emph{no} HILL entropy (equivalently, one can efficiently check if a guess is correct). We leave it as an intriguing open problem to 
overcome these restrictions or to prove they're inherent.

\end{abstract}
\section{Introduction}

Key-derivation considers the  following fundamental problem: 
Given a joint distribution $(X,Z)$ where $X|Z$ (which is short for ``$X$ conditioned on $Z$") is guaranteed to have some kind of entropy, 
derive a ``good" key $K=h(X,S)$ from $X$ 
by means of some efficient key-derivation function $h$, possibly using public randomness $S$. 

In practice, one often uses a cryptographic hash function like SHA3 as the key derivation function $h(.)$  \cite{C:Krawczyk10,C:DGHKR04}, and then simply assumes that $h(.)$ behaves like a random oracle~\cite{CCS:BelRog93}.

In this paper we continue the investigation of key-derivation with provable security guarantees, where we don't make any computational assumption about $h(.)$. 
This problem is fairly well understood for sources $X|Z$ that have high min-entropy (we'll formally define all the entropy notions used in \ref{S:entropy} below), or 
are computationally indistinguishable from having so (in this case, we say $X|Z$ has high HILL entropy). 
In the case where $X|Z$ has $k$ bits of min-entropy, we can either use a strong extractor to derive a $k-2\log\eps^{-1}$ key that is $\eps$-close to uniform, or a condenser to get a $k$ bit key which is $\eps$-close to a variable with $k-\log\log\eps^{-1}$ bits of min-entropy. Using 
extractors/condensers like this also works for HILL entropy, 
except that now we only get computational guarantees (pseudorandom/high HILL entropy) on the derived key.

Often one has to derive a key from a source $X|Z$ which has no HILL entropy at all. The weakest assumption we can make on $X|Z$ for any kind of key-derivation to be possible, is that $X$ is hard to predict given $Z$. 
This has been formalized in \cite{hslure07} by saying that $X|Z$ 
has $k$ bits of unpredictability entropy, denoted $\Uentr{s}(X|Z)\ge k$, if no 
circuit of size $s$ can predict $X$ given $Z$ with advantage $\ge 2^{-k}$ 
(to be more general, we allow an additional parameter $\delta\ge 0$, and $\Uentr{\delta,s}(X|Z)\ge k$ holds if $(X,Z)$ is 
$\delta$-close to some distribution $(Y,Z)$ with $\Uentr{s}(Y|Z)\ge k$). 
We will also consider a more restricted notion, where we say that 
$X|Z$ has $k$ bits of \emph{list}-unpredictability entropy, denoted 
$\UentrX{s}(X|Z)\ge k$, if it has $k$ bits of unpredictability entropy relative to an oracle $\Eq$ which can be used to verify the correct guess ($\Eq$ outputs $1$ on input $X$, and $0$ otherwise).\footnote{We chose this name as having access to $\Eq$ is equivalent to being allowed to output a list of guesses. This is very similar 
to the well known concept of  list-decoding.}
We'll discuss  this notion in more detail below. 
For now, let us just 
mention that for the important special case where it's easy to verify if a guess for $X$ is correct (say, because we condition on 
$Z=f(X)$ for some one-way function\footnote{To be precise, this only holds for \emph{injective} one-way functions. One can generalise list-unpredictability and let $\Eq$ output $1$ on some 
set $\cal X$, and the adversary wins if she outputs any $X\in\cal X$. 
Our results (in particular Theorem~\ref{T:GLC}) also hold for this more general notion, which captures general one-way functions by letting ${\cal X}=f^{-1}(f(X))$ be the set 
of all preimages of $Z=f(X)$.}  $f$), the oracle $\Eq$ does not help, and thus unpredictability and list-unpredictability coincide. 
The results proven in this paper imply 
 that from a source $X|Z$ with $k$ bits of list-unpredictability entropy, it's possible to extract a 
$k$ bit key with $k-3$ bits of HILL entropy
\begin{proposition}
\label{P:main}
Consider a joint distribution $(X,Z)$ over $\bin^n\times\bin^m$ where 
\begin{equation}
\label{e:P1}
\UentrX{s,\gamma}(X|Z)\ge k
\end{equation}
Let 
$S\in \bin^{n\times k}$ be uniformly random and $K=X^TS\in\bin^k$, then the unpredictability entropy of $K$ is
\begin{equation}
\label{e:P2}
\Uentr{s/2^{2k}poly(m,n),\gamma}(K|Z,S)\ge k-3
\end{equation}
and the HILL entropy of $K$ is
\begin{equation}
\label{e:P3}
\hillentr{t,\eps+\gamma}(K|Z,S)\ge k-3
\end{equation}
with\footnote{We denote with $poly(m,n)$ some fixed polynomial in $(n,m)$, but it can denote different polynomial throughout the paper. In particular, the $poly$ here is not the same as in (\ref{e:P2}) as it hides several extra terms.} 
$t
=s\cdot \frac{\eps^7}{2^{2k}poly(m,n)}$.
\end{proposition}
Proposition~\ref{P:main} follows from two results we prove in this paper.  

First, in Section~\ref{S:CU} we prove  Theorem~\ref{T:GLC} which shows how to  ``abuse'' Goldreich-Levin 
hardcore bits by generating a $k$ bit key $K=X^TS$ from a source $X|Z$ with $k$ bits of 
list-unpredictability. The Goldreich-Levin theorem \cite{STOC:GolLev89} implies nothing about the pseudorandomness of $K|(Z,S)$ when extracting that many bits. Instead, we prove that GL is a good ``condenser" for unpredictability entropy: if $X|Z$ has $k$ bits of list-unpredictability entropy, then $K|(Z,S)$ has $k-3$ bits of unpredictability entropy (note that 
we start with list-unpredictability, but only end up with ``normal" unpredictability entropy). 
This result is used in the first step in Proposition~\ref{P:main}, showing that (\ref{e:P1}) implies (\ref{e:P2}).

Second, in Section~\ref{S:H2M} we prove our main result, Theorem~\ref{T:m2} which states 
that any source $X|Z$ which has $|X|-d$ bits of unpredictability entropy, has the same amount of HILL entropy (technically, we show that it implies the same amount of 
metric entropy against deterministic real-valued distinguishers. This notion implies the same amount of HILL entropy as shown by Barak et al. \cite{bashwi03}). The security loss in this argument is exponential in the entropy gap $d$. Thus, if $d$ is very large, this argument is useless, but if we first condense unpredictability as just explained, we have a gap of only $d=3$. 
This result is used in the second step in Proposition~\ref{P:main}, showing that  (\ref{e:P2}) implies (\ref{e:P3}).
In the two sections below we discuss two shortcomings of Theorem~\ref{T:GLC} which we hope can be overcome in future work.\footnote{
After announcing this result at a workshop, we learned that Colin Jia Zheng proved a weaker version of this result. Theorem~4.18 in this PhD thesis, which is available via 
\url{http://dash.harvard.edu/handle/1/11745716}
also states that $k$ bits of unpredictability imply 
$k$ bits of HILL entropy. Like in our case, the loss in circuit size in his proof is 
polynomial in $\eps^{-1}$, but it's also exponential in 
$n$ (the length of $X$), whereas our loss is only exponential in the 
entropy gap $\Delta=n-k$.}
\subsubsection{On the dependency on $2^k$ in Theorem~\ref{T:GLC}.}
As outlined above, our first result is Theorem~\ref{T:GLC}, which shows how to 
condense a source with $k$ bits of list-unpredictability into a $k$ bit key having $k-3$ bits of unpredictability entropy. 
The loss in circuit size is $2^{2k}poly(m,n)$, and it's not clear if the dependency on $2^k$ is necessary here, or if one 
can replace the dependency on $2^k$ with a dependency on $poly(\eps^{-1})$ at the price of an extra $\eps$ term in the distinguishing advantage. 
In many settings $\log(\eps^{-1})$ is in the order of $k$, in which case the above difference is not too important. This is for example the case when considering a $k$ bit 
key for a symmetric primitive like a block-cipher, where one typically assumes the hardness of the cipher to be exponential in the key-length (and thus, if we want 
$\eps$ to be in the same order, we have $\log(\eps^{-1})=\Theta(k)$). In other settings, 
$k$ can be superlinear in $\log(\eps^{-1})$, e.g., if the the high entropy string is used to generate an RSA key.

\subsubsection{List vs. normal Unpredictability.}
Our Theorem~\ref{T:GLC} shows how to condense a source where $X|Z$ has $k$ bits of \emph{list}-unpredictability entropy into a $k$ bit string with $k-3$ bits unpredictability entropy. 
It's an open question to which extent it's necessary to assume \emph{list}-unpredictability here, maybe ``normal" unpredictability is already sufficient? 
Note that list-unpredictability is a lower bound for unpredictability as one always can ignore the 
$\Eq$ oracle, i.e., $\Uentr{\eps,s}(X|Z)\ge \UentrX{\eps,s}(X|Z)$, and in general, list-unpredictability can be much smaller than unpredictability entropy.\footnote{E.g., let $X$ by uniform over $\bin^n$ and $Z$ arbitrary, but independent of $X$, then for $s=\exp(n)$ we have $\Uentr{s}(X|Z)=n$ but 
$\UentrX{s}(X|Z)=0$ as we can simply invoke $\Eq$ on all $\bin^n$ until $X$ is found.}

Interestingly, we can derive a $k$ bit key with almost $k$ bits of HILL entropy from a source $X|Z$ which $k$ bits unpredictability entropy $\Uentr{\eps,s}(X|Z)\ge k$ in 
two extreme cases, namely, if either
\vspace{-0.2cm}
\begin{enumerate}
\item
if $X|Z$ has  basically no HILL entropy (even against small circuits).
\item
or when $X|Z$ has (almost) $k$ bits of (high quality) HILL entropy. 
\end{enumerate}
\vspace{-0.2cm}
In case 1. we observe that if  $\hillentr{\eps,t}(X|Z)\approx 0$ for some $t\ll s$, or equivalently, given $Z$ we can efficiently distinguish $X$ from any $X'\neq X$, then the
$\Eq$ oracle used in the definition of list-unpredictability can be efficiently emulated, which means it's redundant, and thus 
$X|Z$ has the same amount of list-unpredictability and unpredictability entropy, $\Uentr{s,\eps}(X|Z)\approx \UentrX{s',\eps'}(X|Z)$ for $(\eps',s')\approx (\eps,s)$.
Thus, we can use Theorem 1 to derive a $k$ bit key with $k-O(1)$ bits of HILL entropy in this case.
In case 2., we can simply use any  condenser for min-entropy to get a key with HILL entropy $k-\log\log\eps^{-1}$ (cf. Figure~\ref{fig2}). 
As condensing almost all the unpredictability entropy into HILL entropy is possible in the two extreme cases where $X|Z$ has either no or a lot of HILL entropy, it seems conceivable that it's also possible in all the in-between cases (i.e., without making any additional assumptions about $X|Z$ at all).

\subsubsection{GL vs. Condensing.}
Let us stress as this point that, because of the two issues discussed above, our result does not always allow to generate more bits with high 
HILL entropy than just using the Goldreich-Levin theorem. Assuming $k$ bits of unpredictability 
we get $k-3$ of HILL, whereas GL will only give $k-2\log(1/\eps)$. But as currently our reduction has a quantitatively 
larger loss in circuit size than the GL theorem, in order to get HILL entropy of the same quality (i.e., secure against $(s,\delta)$ adversaries for some fixed $(s,\delta)$) we must consider the unpredictability entropy of the source $X|Z$ against 
more powerful adversaries than if we're about to use GL. And in general, the amount of unpredictability 
(or any other computational) entropy of $X|Z$ can decrease as we consider more powerful adversaries.

\section{Entropy Notions}
\newcommand{\dist}{{\cal D}}
\label{S:entropy}
In this section we formally define the different entropy notions considered in this paper. 
We denote with $\dist_s^{rand,\{0,1\}}$ the set of all \emph{probabilistic} circuits of size $s$ with \emph{boolean} output, and $\dist_s^{rand,[0,1]}$ denotes the set of all \emph{probabilistic} 
circuits with \emph{real-valued} output in the range $[0,1]$. The analogous \emph{deterministic}
circuits are denoted $\dist_s^{det,\{0,1\}}$ and $\dist_s^{det,[0,1]}$. 
We use $X\sim_{\eps,s}Y$ to denote computational indistinguishability of variables $X$ and $Y$, formally\footnote{Let us mention that the choice of the distinguisher 
class in (\ref{e:D}) irrelevant (up to a small additive difference in circuit size), we can replace  $\dist_s^{rand,\{0,1\}}$ with any of the three other distinguisher classes.}
\begin{equation}
\label{e:D}
X\sim_{\eps,s}Y\iff \forall \algC\in\dist_s^{rand,\{0,1\}}\ :\ |\Pr[\algC(X)=1]-\Pr[\algC(Y)=1]|\le \eps
\end{equation}
$X\sim_{\eps}Y$ denotes that $X$ and $Y$ have statistical distance $\eps$, i.e., $X\sim_{\eps,\infty}Y$, and 
with $X\sim Y$ we denote that they're identically distributed. With $U_n$ we denote the uniform distribution over
$\bin^n$. 

\begin{definition}\label{def:min_entropy}
  The {\bf min-entropy} of a random variable $X$ with support $\cal
  X$ is 
  $$\minentr(X)=-\log_2\max_{x\in\cal X}\Pr[X=x]$$
  For a pair $(X,Z)$ of random variables, the {\bf average min-entropy} of $X$ conditioned on $Z$ is
$$
    \avminentr(X|Z) \ =\ -\log_2\expt{z\leftarrow Z}\max_x\Pr[X=x|Z=z]
                    \ =\  -\log_2\expt{z\leftarrow Z}2^{-\minentr(X|Z=z)} 
$$
\end{definition}
HILL entropy is a computational variant of min-entropy, where $X$ (conditioned on $Z$) has 
$k$ bits of HILL entropy, if it cannot be distinguished from some $Y$ that (conditioned on $Z$) has $k$ bits of min-entropy, formally
\begin{definition}[\cite{HILL99},\cite{hslure07}]\label{def:hill}
  A random variable $X$ has {\bf HILL entropy} $k$, denoted by $\hillentr{\varepsilon,s}(X)\geq k$, if there exists a distribution $Y$ satisfying $\minentr(Y)\geq k$ and $X\epsclose{\varepsilon,s} Y$.

Let $(X,Z)$ be a joint distribution of random variables. Then $X$ has
{\bf conditional HILL entropy} $k$
conditioned on $Z$, denoted by $\hillentr{\varepsilon,s}(X|Z)\geq k$, if there exists 
a joint distribution $(Y,Z)$ such that $\avminentr(Y|Z)\geq k$ and $(X,Z)\epsclose{\varepsilon,s} (Y,Z)$.
\end{definition}
Barak, Sahaltiel and Wigderson \cite{bashwi03} define the notion of 
metric entropy, which is defined like HILL, but the quantifiers are exchanged. That is, instead of asking for 
a single distribution $(Y,Z)$ that fools all distinguishers, we only ask that for every 
distinguisher $\algD$, there exists such a distribution. For reasons discussed in 
Section~\ref{S:MH}, in the definition below we make the class of distinguishers considered explicit.
\begin{definition}[\cite{bashwi03},\cite{cryptoeprint:2012:466}]
\label{def:metric}
Let $(X,Z)$ be a joint distribution of random variables. 
Then $X$ has {\bf conditional metric entropy} $k$  conditioned on $Z$ (against probabilistic boolean distinguishers), denoted by 
 $\metricentrrb{\varepsilon,s}(X|Z)\geq k$, if for every $\algD\in\dist_s^{rand,\{0,1\}}$ 
 there exists a joint distribution $(Y,Z)$ such that $\avminentr(Y|Z)\geq k$ and 
 $$
 |\Pr[\algD(X,Z)=1]-\Pr[\algD(Y,Z)=1]|\le \eps
 $$
 More generally, for $class\in\{rand,det\},range\in\{[0,1],\{0,1\}\}$, \\
 $\metricentrc{\varepsilon,s}(X|Z)\geq k$ if for every 
 $\algD\in\dist_s^{class,range}$ such a $(Y,Z)$ exists.
\end{definition}
Like HILL entropy, also 
unpredictability entropy, which we'll define next, can be seen as a computational variant of 
min-entropy. Here  we don't require indistinguishability as for HILL entropy, but only that 
the variable is hard to predict.
\begin{definition}[\cite{hslure07}]\label{def:uentropy}
$X$ has {\bf unpredictability entropy} $k$
conditioned on $Z$, denoted by $\Uentr{\eps,s}(X|Z)\geq k$, if  
$(X,Z)$ is $(\eps,s)$ indistinguishable from some $(Y,Z)$, where no probabilistic circuit of size 
$s$ can predict $Y$ given $Z$ with probability better than $2^{-k}$, i.e.,
\begin{equation}
\label{e:C}
\Uentr{s,\eps}(X|Z)\geq k
\iff
 \exists (Y,Z),(X,Z)\epsclose{\varepsilon,s} (Y,Z)\ \forall \algC,|\algC|\le s\ :\  \!\!\!\!\!\!\!\!\!\!\!
\Pr_{(y,z)\gets (Y,Z)}\!\!\!\!\!\!\!\!\![\algC(z)=y]\le 2^{-k}
\end{equation}
We also define a notion called ``list-unpredictability'', denoted 
$\UentrX{\eps,s}(X|Z)\geq k$, which holds if 
$\Uentr{\eps,s}(X|Z)\geq k$ as in (\ref{e:C}), but where $\algC$ additionally gets oracle access to a function $\Eq(.)$ which outputs $1$ on input $y$ and $0$ otherwise. So, $\algC$ 
can efficiently test if some candidate guess for $y$ is correct.\footnote{We name this notion 
"list-unpredictability" as we get the same notion when instead of giving $\algC$ oracle access to $\Eq(.)$, we allow $\algC(z)$ to output a list of guesses for $y$, not just one value, and require that $\Pr_{(y,z)\gets (Y,Z)}[y\in \algC(z)]\le 2^{-k}$. This notion 
is inspired by the well known notion of list-decoding.}

\end{definition}
\begin{remark}[The $\eps$ parameter]
\label{R:1}
The $\eps$ parameter in the definition above is not really necessary, following \cite{EC:HsiLuRey07},
we added it so we can have a ``smooth" notion, 
which is easier to compare to HILL or smooth min-entropy. 
If $\eps=0$, we'll simply omit it, then the definition simplifies to 
$$\Uentr{s}(X|Z)\geq k \iff \Pr_{(x,z)\gets (X,Z)}[\algC(z)=x]\le 2^{-k}$$
Let us also mention that unpredictability entropy is only interesting if the conditional part $Z$ is not empty as 
(already for $s$ that is linear in the length of $X$) we have $\Uentr{s}(X)=\minentr(X)$ which can be seen by considering the circuit 
$\algC$ (that gets no input as $Z$ is empty) which simply outputs the constant $x$ maximizing $\Pr[X=x]$.
\end{remark}
\subsubsection{Metric vs. HILL.}
\label{S:MH}
We will use a lemma which states that deterministic real-valued metric entropy implies 
the same amount of HILL entropy (albeit, with some loss in quality). 
This lemma has been proven by \cite{bashwi03} for the unconditional case, i.e., when 
$Z$ in the lemma below is empty, it has been observed by \cite{cryptoeprint:2012:466,C:CKLR11} that the proof also holds in the conditional case as stated below
\begin{lemma}[\cite{bashwi03,cryptoeprint:2012:466,C:CKLR11}]
\label{L:M2H}
For any joint distribution $(X,Z)\in\bin^n\times\bin^m$ and any $\eps,\delta,k,s$
$$
\metricentrdr{\eps,s}(X|Z)\ge k\quad \Rightarrow \quad
\hillentr{\eps+\delta, s\cdot \delta^2/(m+n)}(X|Z)\ge k
$$
\end{lemma}
Note that in Definition~\ref{def:hill} of HILL entropy, we only consider security against 
probabilistic boolean distinguishers (as $\sim_{\eps,s}$ was defined this way), whereas in Definiton~\ref{def:metric} of metric entropy  
we make the class of distinguishers explicit. The reason for this is that in the definition of 
HILL entropy the class of distinguishers considered is irrelevant 
(except for a small additive degradation in circuit size, cf. \cite[Lemma 2.1]{cryptoeprint:2012:466}).\footnote{This 
easily follows from the fact  that in the definition (\ref{e:D}) of  computational indistinguishability  the choice of the distinguisher class is irrelevant.}
Unlike for HILL, for  metric entropy the choice of the distinguisher class does matter. 
In particular, deterministic boolean metric entropy $\metricentrdb{\eps,s}(X|Y)\ge k$ is only known to imply deterministic real-valued metric entropy $\metricentrdr{\eps+\delta,s}(X|Y)\ge k-\log(\delta^{-1})$, i.e., we must allow for a $\delta>0$ loss in distinguishing advantage, and this will at the same time result in a loss of $\log(\delta^{-1})$ in the amount of entropy. 
For this reason, it is crucial that in Theorem~\ref{T:m2} we show 
that unpredictability entropy implies deterministic \emph{real-valued} metric entropy, 
so we can then apply Lemma~\ref{L:M2H} to get the same amount of HILL entropy. 
Dealing with real-valued distinguishers is the main source of technical difficulty in 
the proof of the Theorem~\ref{T:m2}, proving the analogous statement for deterministic \emph{boolean} distinguishers is much simpler.
\section{Known Results on Provably Secure Key-Derivation}

We say that a cryptographic scheme has security $\alpha$, if no adversary (from some 
class of adversaries like all polynomial size circuits) can win some security game with advantage $\ge \alpha$  
if the scheme is
instantiated with a uniformly random string.\footnote{We'll call this string ``key". Though in many settings (in particular when keys are not simply uniform random strings, like in public-key crypto) this string is not used as a key directly, but one rather should think of it as the randomness used to sample the actual keys.} 
Below we will distinguish between \emph{unpredictability} applications,  
where the advantage bounds the probability of winning some security game (a typical example are digital signature schemes, where the game captures the existential unforgeability under chosen message attacks), and \emph{indistinguishability} applications, where the advantage bounds the distinguishing advantage from some ideal object 
(a typical example is the security definition of pseudorandom generators or functions).
\subsection{Key-Derivation from Min-Entropy}
\paragraph{Strong Extractors.}
Let $(X,Z)$ be a source where 
$\avminentr(X|Z)\ge k$, or equivalently, no adversary can guess $X$ given $Z$ with probability better than $2^{-k}$
(cf. Def.~\ref{def:min_entropy}). 
Consider the case where we want to derive a key 
$K=h(X,S)$ that is statistically close to uniform given $(Z,S)$. For example, $X$ could be some physical source (like statistics from keystrokes) from which we want to generate almost uniform randomness. Here  
$Z$ models potential side-information the adversary might have on $X$.
This setting is very well understood, and such a key can be derived using a strong extractor as defined below.
\begin{definition}[\cite{STOC:NisZuc93},\cite{dors10}]
A function $\ext:\bin^n\times\bin^d\rightarrow \bin^\ell$
is an average-case $(k,\eps)$-strong extractor if for
every distribution $(X,Z)$ over $\bin^n\times \bin^m$
with $\avminentr(X|Z)\ge k$ and $S\sim U_d$, the distribution 
$(\ext(X,S),S,Z)$ has statistical distance $\eps$ to $(U_\ell,S,Z)$. 
\end{definition}
Extractors $\ext$ as above exist with $\ell=k-2\log(1/\eps)$  \cite{HILL99}. 
Thus, from any $(X,Z)$ where $\avminentr(X|Z)\ge k$ we can extract a key 
 $K=\ext(X,S)$ of length 
$k-2\log(1/\epsilon)$ that is $\epsilon$ close to uniform \cite{HILL99}. 
The entropy gap $2\log(1/\eps)$ is optimal by the so called ``RT-bound" \cite{RT00}, even if 
we assume the source is efficiently samplable \cite{EC:DPW14}. 

If instead of using a uniform $\ell$ bit key for an  $\alpha$ secure scheme, we use a key
that is  $\eps$ close to uniform, the scheme will still be at least $\beta=\alpha+\eps$ secure. 
In order to get security $\beta$ that is of the same order as $\alpha$, we thus must set $\eps\approx \alpha$. 
When  the available amount $k$ of min-entropy is small, for example when dealing with biometric data \cite{dors10,EC:BDKOS05}, a loss of $2\log(1/\epsilon)$  bits (that's $160$ bits for a typical security level $\eps=2^{-80}$) is often unacceptable. 
\paragraph{Condensers.}
The above bound is basically tight for many
\emph{indistinguishability} 
applications like pseudorandom generators or pseudorandom functions.\footnote{For example, consider a  pseudorandom function $\F:\bin^k\times\bin^a\rightarrow\bin$
and a key $K$ that is uniform over all 
keys where $\F(K,0)=0$, this distribution is $\eps\approx 1/2$ close to uniform and has min-entropy $\approx |K|-1$, but the security breaks completely as one can distinguish $\F(U_k,.)$ from $\F(K,.)$ with advantage $\beta\approx 1/2$ (by quering on input $0$, and outputting $1$ iff the output is $0$).}
Fortunately, for many applications a close to uniform key is not necessary, and a key $|K|$ with min-entropy 
$|K|-\Delta$ for some small $\Delta$ is basically as good as a uniform one. This is the case for all \emph{unpredictability} applications, which includes OWFs, digital-signatures and MACs.\footnote{\cite{TCC:DodYu13} identify an interesting class of applications called ``square-friendly", this class contains all unpredictability applications, and some indistinguishability applications like weak PRFs (which are PRFs that can only be queried on random inputs). This class of applications remains somewhat secure even for a small entropy gap $\Delta$: For $\Delta=1$ the security is  $\beta\approx\sqrt{\alpha}$. This is worse that the $\beta=2\alpha$ for unpredictability applications, but much better than the complete loss of security 
$\beta\approx 1/2$ required for some indistinguishability apps like (standard) PRFs.}
It's not hard to show that if the scheme is $\alpha$ secure with a uniform key 
it remains at least $\beta=\alpha  2^\Delta$ secure (against the same class of attackers) if instantiated with any key $K$ that has $|K|-\Delta$ bits of min-entropy.\footnote{Assume some adversary breaks the scheme, say, forges a signature, with advantage $\beta$ if the key comes from the distribution $K$. If we sample a uniform 
key instead, it will have the same distribution as $K$ conditioned on an event that holds with probability $2^{-\Delta}$, and thus this adversary will still break the scheme with 
probability $\beta/2^{\Delta}$.}
Thus, for unpredictability applications we don't have to extract an almost uniform key, but ``condensing" $X$ into a key with $|K|-\Delta$ bits of min-entropy for some small $\Delta$ is enough.

\cite{EC:DPW14} show that a $(\log\eps+1)$-wise independent hash function 
$\cond:\bin^n\times\bin^d\rightarrow\bin^\ell$ is a condenser with the following parameters. 
For any $(X,Z)$ where $\avminentr(X|Z)\ge \ell$, for a random seed $S$ (used 
to sample a $(\log\eps+1)$-wise independent hash function), the distribution 
$(\cond(X,S),S)$ is $\eps$ close to a distribution $(Y,S)$ where $\avminentr(Y|Z)\ge \ell-\log\log(1/\eps)$. 
Using such an $\ell$ bit key (condensed from a source with $\ell$ bits min-entropy) for an unpredictability application that is $\alpha$ secure (when using a uniform $\ell$ bit key), we get security
$\beta\le \alpha2^{\log\log(1/\eps)}+\eps$, which setting $\eps=\alpha$ gives 
$\beta \le \alpha(1+\log(1/\alpha))$ security,   thus, security degrades only by a logarithmic factor.

%

\subsection{Key-Derivation from Computational Entropy}
The bounds discussed in this section are summarised in Figures \ref{fig1} and \ref{fig2} in Appendix~\ref{S:fig}. The last row of 
Figure~\ref{fig2} is the new result proven in this paper.
\paragraph{HILL Entropy.}
As already discussed in the introduction, often we want to derive a key from a distribution $(X,Z)$ where there's no ``real" min-entropy at all 
$\avminentr(X|Z)=0$. 
This is for example the case when $Z$ is the transcript (that can be observed by an adversary) of a key-exchange protocol like Diffie-Hellman, where the agreed value $X=g^{ab}$ is determined by the transcript $Z=(g^a,g^b)$ \cite{C:Krawczyk10,EC:GenKraRab04}. Another setting where this can be the case is in the context of side-channel attacks, where the leakage $Z$ from a device can completely determine its internal state $X$.  

If $X|Z$ has $k$ bits of HILL entropy, i.e., is computationally indistinguishable from having min-entropy $k$ (cf. Def.~\ref{def:hill}) 
we can derive keys exactly as described above assuming $X|Z$ had $k$ bits of min-entropy. In particular, if $X|Z$ has $|K|+2\log(1/\eps)$ bits of HILL entropy for some negligible $\eps$, we can derive a key $K$ that is pseudorandom,  and if $X|Z$ has  
$|K|+\log\log(1/\eps)$ bits of HILL entropy, we can derive a key that is almost as good as a uniform one for any unpredictability application.

\paragraph{Unpredictability Entropy.}
Clearly, the minimal assumption we must make on a distribution $(X,Z)\in\bin^n\times\bin^m$ for any 
key derivation to be possible at all is that $X$ is hard to compute given $Z$, that is, $X|Z$ must have some unpredictability entropy as in Definition~\ref{def:uentropy}. 
Goldreich and Levin \cite{STOC:GolLev89} show how to generate pseudorandom bits 
from such a source. In particular, the Goldreich-Levin theorem implies that if $X|Z$ has 
at least $2\log\eps^{-1}$ bits of list-unpredictability, then the inner product $R^TX$ of $X$ with a random vector $R$ is $\eps$ indistinguishable from uniformly random (the loss in circuit size is 
$poly(n,m)/\eps^4$). Using the chain rule for unpredictability entropy,\footnote{Which states that if 
$X|Z$ has $k$ bits of list-unpredictability, then for any $(A,R)$ where $R$ is independent of $(X,Z)$, 
$X|(Z,A,R)$ has $k-|A|$ bits of list-unpredictability entropy. In particular, extracting $\ell$ 
inner product bits, decreases the list-unpredictability by at most $\ell$.} we can generate an
$\ell=k-2\log\eps^{-1}$ bit long pseudorandom string that is $\ell \eps$ 
indistinguishable (the extra $\ell$ factor comes from taking the union bound over all bits) from uniform.

Thus, we can turn $k$ bits of list-unpredictability into $k-2\log\eps^{-1}$ bits of pseudorandom 
bits (and thus also that much HILL entropy) with quality roughly $\eps$. The question whether it's possible to generate significantly more than $k-2\log\eps^{-1}$ of HILL entropy from a 
source with $k$ bits of (list-)unpredictability seems to have never been addressed in the literature before. The reason might be that one usually is interested in generating pseudorandom bits (not just HILL entropy), and for this, the $2\log\eps^{-1}$ entropy loss is inherent. The observation that for many applications high HILL entropy is basically as good as pseudorandomness is more recent, and recently gained 
attention by its usefulness in the context of leakage-resilient cryptography \cite{FOCS:DziPie08,TCC:DodYu13}. 

In this paper we prove that it's in fact possible to turn almost all list-unpredictability into HILL entropy. 

\section{Condensing Unpredictability}
\label{S:CU}
Below we state Theorem~\ref{T:GLC} whose proof is in Appendix~\ref{A:PT2}, but first, let us give some intuition. Let $X|Z$ have $k$ bits of list-unpredictability, and assume we start extracting Goldreich-Levin hardcore bits $A_1,A_2,\ldots$ 
by taking inner products $A_i=R_i^TX$ for random $R_i$. The first extracted bits $A_1,A_2,\ldots$ will be pseudorandom (given the $R_i$ and $Z$), but with every extracted bit, the list-unpredictability can also decrease by one bit. As the GL theorem requires at least 
$2\log\eps^{-1}$ bits of list-unpredictability to extract an $\eps$ secure pseudorandom bit, we must stop after 
$k-2\log\eps^{-1}$ bits. In particular, the more we extract, the worse the pseudorandomness of the extracted string becomes. Unlike the original GL theorem, in our Theorem~\ref{T:GLC} we only argue about the unpredictability of the extracted string, and unpredictability entropy has the nice property that it can never decrease, i.e., predicting $A_1,\ldots,A_{i+1}$ is always at least as hard as predicting $A_1,\ldots,A_{i}$. Thus, despite the fact that once $i$ approaches $k$ it becomes easier and easier to predict  $A_i$ 
(given $A_1,\ldots,A_{i-1},Z$ and the $R_i$'s)\footnote{The only thing we know about the last extracted bit $A_k$ is that it cannot be predicted with advantage $\ge 0.75$, more generally, $A_{k-j}$ cannot be predicted with 
advantage $1/2+1/2^{j+2}$.} 
 this hardness will still add up to $k-O(1)$ bits of unpredictability entropy. 
 
The proof is by contradiction, we assume that $A_1,\ldots,A_k$ can be predicted with advantage $2^{-k+3}$ (i.e., does not have $k-3$ bits of unpredictability), and then use such a predictor to predict $X$ with advantage $>2^{-k}$, contradicting the $k$ bit list-unpredictability of $X|Z$. 

If $A_1,\ldots,A_k$ can be predicted as above, then there must be an index $j$ s.t. 
$A_j$ can be predicted with good probability conditioned on $A_1,\ldots,A_{j-1}$ being correctly predicted. 
We then can use the Goldreich-Levin theorem, which tells us how to find $X$ given such a predictor. 
Unfortunately, $j$ can be close to $k$, and to apply the GL theorem, we first need to find the right
values for $A_1,\ldots,A_{j-1}$ on which we condition, and also can only use the predictor's guess 
for $A_j$ if it was correct on the first $j-1$ bits. We have no better strategy for this than trying all possible values, and this is the reason why the loss in circuit size in Theorem~\ref{T:GLC} depends on $2^k$. 

In our proof, instead of using the Goldreich-Levin theorem, we will actually use a more fine-grained variant due to Hast which allows to distinguish between errors and erasures (i.e., cases where we know that we don't have
any good guess. As outlined above, this will be the case whenever the predictor's guess for the 
first $j-1$ inner products was wrong, and thus we can't assume anything about the $j$th guess being correct). This will give a much better quantitative bound than what seems possible using GL. 
\begin{theorem}[Condensing Upredictability Entropy]
\label{T:GLC}
Consider any distribution $(X,Z)$ over $\bin^n\times\bin^m$ where
$$
\UentrX{\eps,s}(X|Z)\ge k
$$
then for a random $R\gets \bin^{k\times n}$
$$\Uentr{\eps,t}(R.X|Z,R)\ge k-\Delta$$
where\footnote{We can set $\Delta$ to be any constant $>1$ here, but choosing a smaller $\Delta$ would imply a smaller $t$.}
$$
t=\frac{s}{2^{2k}\operatorname{poly}(m,n)}
\qquad,
\qquad
\Delta=3
$$
\end{theorem}

\section{High Unpredictability implies Metric Entropy}
\label{S:H2M}
In this section we state our main results, showing that $k$ bits of unpredictability entropy imply the same amount 
of HILL entropy, with a loss exponential in the ``entropy gap". The proof is in Appendix~\ref{A:Pmain}.
\begin{theorem}[Unpredictability Entropy Implies HILL Entropy]
\label{T:m2}
For any distribution $(X,Z)$ over $\bin^n\times\bin^m$, if $X|Z$ has unpredictability entropy
\begin{equation}
\label{e:imp0}
\Uentr{\gamma,s}(X|Z)\ge k
\end{equation}
 then, with $\Delta=n-k$ denoting the entropy gap, $X|Z$ has (real valued, deterministic) metric entropy
\begin{equation}
\label{e:important}
\metricentrdr{\eps+\gamma,t}(X|Z)\ge k
\quad\textrm{for}\quad t = \Omega\left( s \cdot \frac{ \eps^5}{2^{5\Delta}\log^2\left( 2^{\Delta}\epsilon^{-1} \right)}\right)
\end{equation}
By Lemma~\ref{L:M2H} this further implies
that $X|Z$ has, for any $\delta>0$, HILL entropy 
$$
\hillentr{\eps+\delta+\gamma,\Omega(t\delta^2/(n+m))}(X|Z)\ge k
$$
which for $\eps=\delta=\gamma$ is
$$
\hillentr{3\eps,\Omega(s\cdot \eps^7/2^{5\Delta}(n+m)\log^2\left( 2^{\Delta}\epsilon^{-1} \right))}(X|Z)\ge k
$$
\end{theorem}

\bibliography{abbrev3,crypto,pseudoentropy}
\bibliographystyle{alpha}
\addcontentsline{toc}{section}{References}

\appendix
\newpage
\section{Figures}
\label{S:fig}
\begin{figure}[h]
\begin{center}
\scalebox{0.8}{
\begin{tabular}{|l|c|c|c|c|c|c|}
\hline
\multicolumn{4}{|c|}{Deriving a (pseudo)random key of length $|K|=k-2\log\eps^{-1}$}
\\
\multicolumn{4}{|c|}{from a source $(X,Z)\in\bin^n\times\bin^m$ where $X|Z$ has $k$ bits (min/HILL/list-unpredictability) entropy}
\\
\hline
Entropy&Entropy quantity and &Derive key $K$ of&
Quality of derived key
\\
type&quality of source&length $k-2\log\eps^{-1}$ as&
$\hillentr{\eps',s'}(K|Z,S)=k-2\log \eps^{-1}=|K|$
\\
&&&equivalently\\
&&&$(K,Z,S)\sim_{\eps',s'}(U_{|K|},Z,S)$\\
\hline
min&
$\avminentr(X|Z)=k$&
$K=\ext(X,S)$&
$\eps'=\eps \qquad s'=\infty$
\\
\hline
HILL&
$\hillentr{\delta,s}(X|Z)=k$&
$K=\ext(X,S)$&
$\eps'=\eps+\delta \qquad s'\approx s$\\
\hline
Unpredict.&
$\UentrX{\delta,s}(X|Z)=k$&
$K=\GL(X,S)=S^TX$&
$
\eps'= m\eps+\delta\qquad s'=s\cdot \eps^4/poly(m,n)
$
\\
\hline
\end{tabular}
}
\end{center}
\caption{
\label{fig1}
Bounds on deriving a (pseudo)random key $K$ of length $|K|=k-2\log \eps^{-1}$ bit 
from a source 
$X|Z$ with $k$ bits  of min, HILL or list-unpredictability entropy. $\ext$ is a 
strong extractor (e.g. leftover hashing), and 
$\GL$ denotes the Goldreich-Levin construction, which for $X\in\bin^n$ and 
$S\in\bin^{n\times |K|}$ is simply defined as  $\GL(X,S)=S^TX$. 
Leftover hashing requires a seed of length $|S|=2n$ (extractors with a much shorter seed 
$|S|=O(\log n+\log \eps^{-1})$ that extract $k-2\log\eps^{-1}-O(1)$ bits also exist), whereas Goldreich-Levin 
requires a longer $|S|=|K|n$ bit seed. 
The above bound for HILL entropy even holds if $X|Z$ only has 
$k$ bits of probabilistic boolean metric entropy (a notion implying the same amount of HILL entropy, albeit with a 
loss in circuit size), as shown in 
Theorem 2.5 of \cite{cryptoeprint:2012:466}
}
\end{figure}
\begin{figure}[h]
\begin{center}
\scalebox{0.8}{
\begin{tabular}{|r|c|c|c|c|}
\hline
\multicolumn{4}{|c|}{Deriving $k$ bit key $K$
with high HILL entropy from 
$X|Z$ with $k$ bits (min/HILL/list-unpredictability) entropy
}\\
\hline
Entropy&Entropy quantity and &Derive key of&Quantity and quality of HILL entropy of $K$
\\
type&quality of soucre&length $|K|=k$ as&
$\hillentr{\eps',s'}(K|Z,S)\ge k-\Delta$
\\
\hline
min&
$\avminentr(X|Z)=k$&
$K=\cond(X,S)$&
$\eps'=\eps
\qquad
s'=\infty\qquad \Delta=\log\log\eps^{-1}$
\\
\hline
HILL&
$\hillentr{\delta,s}(X|Z)=k$&
$K=\cond(X,S)$&
$
\eps'=\eps+\delta
\qquad
s'\approx s\qquad \Delta=\log\log\eps^{-1}$
\\
\hline
Unpredict.&
$\UentrX{\delta,s}(X|Z)=k$&
$K=\GL(X,S)=S^TX$&
$\eps'=\eps+\delta\qquad s'=s\cdot {\eps^7}/{2^{2k}poly(m,n)}\qquad \Delta=3$
\\
\hline
\end{tabular}
}
\end{center}
 \caption{
 \label{fig2}
Bounds on deriving a key of length $k$ with min (or HILL) entropy $k-\Delta$ from a source $X|Z$ with $k$ bits of min, HILL or unpredictability entropy. $\cond$ denotes a 
$(\log\eps+1)$ wise independent hash function, which is 
shown to be a good condenser (as stated in the table) for min-entropy in \cite{EC:DPW14}. The bounds for HILL entropy follow directly from the bound for min-entropy. The last row follows from the results in this paper as stated in Proposition~\ref{P:main}.
}
\end{figure}

\input{CONDproof}
\input{U2Mproof}

\input{U2M2}

\end{document}

%% file: CONDproof.tex
\section{Proof of Theorem~\ref{T:GLC}}
\label{A:PT2}
We will use the following theorem due  Hast \cite{EC:Hast03} on decoding Hadamard code with errors and erasures.
\begin{theorem}[\cite{EC:Hast03}]
\label{thm:erasures}
There is an algorithm {\sf LD} that, on input $l$ and $n$ and with oracle access to a binary Hadamard code of $x$ (where $|x|=n$) with an $e$-fraction of errors and an $s$-fraction of erasures, can output a list of $2^l$ elements in time $O(nl2^l)$ asking $n2^l$ oracle queries such that the probability that $x$ is contained in the list is at least $0.8$ if $l\ge \log_2(20n(e+c)/(c-e)^2+1)$, where $c=1-s-e$ (the fraction of the correct answers from the oracle).
\end{theorem}

We'll often consider sequences $v_1,v_2,\ldots$ of values and will use the notation  
$v_a^b$ to denote $(v_a,\ldots,v_b)$, with $v_a^b=\emptyset$ if $a>b$. 
$v^b$ is short for $v_1^b=(v_1,\ldots,v_b)$.

\begin{proof}[of Theorem~\ref{T:GLC}]
It's sufficient to prove the theorem for $\eps=0$, the general case $\eps\ge 0$ then follows directly by the definition 
of unpredictability entropy. To prove the theorem we'll prove its contraposition
\begin{equation}
\label{e:contra}
\Uentr{t}(R.X|Z,R)< k-\Delta
\quad\Rightarrow\quad
\UentrX{s}(X|Z)< k
\end{equation}
The left-hand side of (\ref{e:contra}) means there exists a circuit $\algA$ of size $|\algA|\le t$ such that
\begin{equation}
\label{e:A}
\Pr_{(x,z)\gets(X,Z),r\gets \bin^{k\times n}}[\algA(z,r)=r.x]\ge 2^{-k+\Delta}
\end{equation}
It will be convenient to assume that $\algA$ initially flips a coin $b$, and if $b=0$ outputs 
a uniformly random guess. This loses at most a factor $2$ in $\algA$'s advantage, i.e.,
\begin{equation}
\label{e:A2}
\Pr_{(x,z)\gets(X,Z),r\gets \bin^{k\times n}}[\algA(z,r)=r.x]\ge 2^{-k+\Delta-1}
\end{equation}
but now we can assume that for any $z,r$ and $w\in\bin^k$
\begin{equation}
\label{e:dummy}
\Pr[\algA(z,r)=w]\ge 2^{-k-1}
\end{equation}
Using Markov eq.(\ref{e:A2}) gives us
\begin{equation}
\label{e:Ma}
\Pr_{(x,z)\gets(X,Z)}[\Pr_{r\gets \bin^{k\times n}}[\algA(z,r)=r.x]\ge 2^{-k+\Delta-2}]\ge
2^{-k+\Delta-2} 
\end{equation}
We call $(x,z)\in\supp{(X,Z)}$ ``good" if 
\begin{equation}
\label{e:good}
(x,z)\textrm{ is good }\iff
\Pr_{r\gets \bin^{k\times n}}[\algA(z,r)=r.x]\ge 2^{-k+\Delta-2}
\end{equation}
Note that by eq.(\ref{e:Ma}), $(z,x)\gets (Z,X)$ is good with probability $\ge 2^{-k+\Delta-2}$. 

We will use $\algA$ to construct a new circuit $\algB$ of size $s=O(t2^{2k}\operatorname{poly}(n))$ 
where 
\begin{equation}
\label{e:B}
\Pr_{(x,z)\gets (X,Z)}[\algB(z)=x\ | (x,z)\textrm{ is good}]> 1/2
\end{equation}
Which with (\ref{e:B}) and (\ref{e:Ma}) further gives 
\begin{eqnarray}
\label{e:B2}
\Pr_{(x,z)\gets (X,Z)}[\algB(z)=x]
&=&
\nonumber
\Pr[\algB(z)=x | (x,z)\textrm{ is good}]
\cdot\Pr[(x,z)\textrm{ is good}]\\
&>& 2^{-1}\cdot 2^{-k+\Delta-2}=2^{-k+\Delta-3}
\end{eqnarray}
contradicting the right-hand side of (\ref{e:contra}), and thus proving the theorem.

We'll now construct $\algB$ satisfying (\ref{e:B}), for this, consider any good $(x,z)$. 
Let $R=R^k=(R_1,\ldots,R_k)$ be uniformly random and let $A=A^k=(A_1,\ldots,A_k)$ where $A_i=R_i.x$. 

Let $\hat A\gets \algA(z,R)$ and define $\epsilon_i=\Pr_{R}[\hat A_i=A_i| \hat A^{i-1}=A^{i-1}]$. Using (\ref{e:good}) in the last step
$$
\prod_{i=1}^k\epsilon_i 
=\Pr_R[A=\hat A]=\Pr_R[\algA(z,R)=R.x]\ge 2^{-k+\Delta-2}
$$
Thus, here exists an $i$ s.t., $\epsilon_i\ge 
2^{\frac{-k+\Delta-2}{k}}=\frac{1}{2}+\delta$ with 
$\delta\approx \frac{\Delta-2}{k}\cdot\frac{\ln(2)}{2}$. We fix this $i$ (we don't know 
which $i$ is good, and later will simply try all of them). Then

$$
\E_{R^{i-1}}[\Pr_{R_i,R_{i+1}^k}[\hat A_i=A_i \ |\ 
\hat A^{i-1}=A^{i-1}
]
]\ge 1/2+\delta
$$
Using Markov
\begin{equation}
\label{e:M}
\Pr_{R^{i-1}}[
\Pr_{R_i,R_{i+1}^k}[\hat A_i=A_i \ |\ 
\hat A^{i-1}=A^{i-1}
]\ge 1/2+\delta/2
]\ge \frac{\delta}{2}
\end{equation}
We call $r^{i-1}$ good if (note that by the previous equation a random $r^{i-1}$ is good with probability $\ge \delta/2$).
\begin{equation}
\label{e:goodr}
r^{i-1}\textrm{ is good }\iff
\Pr_{R_i,R_{i+1}^k}[\hat A_i=A_i \ |\ 
\hat A^{i-1}=A^{i-1}
]\ge 1/2+\delta/2
\end{equation}
From now on, we fix some good $r^{i-1}$ and assume  we know 
$a^{i-1}=r^{i-1}.x$ (later we'll simply try all possible choices for $a^{i-1}$). 

We define a predictor $\algP_i(r_i)$ that
tries to predict $r_i.x$ given a random $r_i$ (and also knows $z,r^{i-1},a^{i-1}$ as above) as 
follows
\begin{enumerate}
\item Sample random $r_{i+1}^k\gets R_{i+1}^k$
\item Invoke $\hat A^k \gets \algA(z,r^{(i)},x)$. Note that $r^{(i)}=(r^{i-1},r_i,r_{i+1}^k)$ consists of the fixed $r^{i-1}$, the input $r_i$ and the randomly sampled $r_{i+1}^k$.
\item
if $\hat A^{i-1}=a^{i-1}$ output $\hat A_i$, otherwise output $\bot$.
\end{enumerate}
Using (\ref{e:dummy}), which implies $\Pr[\hat A^{i-1}=a^{i-1}]\ge 2^{-i}$,  
and (\ref{e:goodr}) we can lower bound $\algP_i$'s rate and advantage as
\begin{eqnarray}
\nonumber
\Pr_{R_i}[\algP_i(R_i)\neq\bot]&=&\Pr[\hat A^{i-1}=a^{i-1}] \ge 2^{-i},\\
\Pr_{R_i}[\algP_i(R_i)=R_i.x]&\ge&\Pr[\hat A^{i-1}=a^{i-1}] (\frac{1}{2}+\delta/2). 
\label{e:Pgood}
\end{eqnarray}

In terms of Theorem~\ref{thm:erasures}, we have a binary Hadamard code with $e+c=\Pr[\hat A^{i-1}=a^{i-1}]$, $c-e=\delta\cdot\Pr[\hat A^{i-1}=a^{i-1}]$, which implies that $(e+c)/(c-e)^2\le\frac{2^i}{\delta^2}$.


Now Theorem~\ref{thm:erasures} implies that given such a predictor $\algP$ we can output a list 
that contains $x$ with probability $>0.8$ in time $O(2^i\operatorname{poly}(m,n))=O(2^k\operatorname{poly}(m,n))$, as we assume access to an oracle $\Eq$ with outputs $1$ on input $x$ and $0$ otherwise, we can find $x$ in this list with the same probability. 


Using this, we can now construct an algorithm as claimed in (\ref{e:B}) as follows: 
$\algB$ will sample $i\in\{1,\ldots,k\}$ and then $r^{i-1}$ at random. Then $\algB$ calls $\algP_i$ with all possible $a^{i-1}\in\{0,1\}^{i-1}$.
We note that with probability $\delta/2k$ (we lose a factor $k$ for the guess of $i$, and $\delta/2$ is the probability of sampling a good $r^{i-1}$) the predictor $\algP_i$  will satisfy 
(\ref{e:Pgood}).

If $x$ is not found, $\algB$ repeats the above process, but stops if $x$ is not found after $2k/\delta$ iterations. The success probability of $\algB$ is $\approx (1-1/e)0.8>0.5$ as claimed, the overall running time 
we get is $O(2^{2k}\operatorname{poly}(m,n))$.
\qed
\end{proof}

%% file: U2Mproof.tex
\section{Proof of Theorem~\ref{T:m2}}
\label{A:Pmain}
\newcommand{\out}[1]{}

\noindent 
It's sufficient to prove the theorem for $\gamma=0$, the case $\gamma>0$ then follows 
directly by definition of unpredictability entropy. 
Suppose for the sake of contradiction that (\ref{e:important}) does not hold. 
That is, $\metricentrdr{t,\epsilon}(X|Z) < k$, which 
means that there exists a distinguisher $\cD:\nstrings\times\mstrings \rightarrow [0,1]$ of size $t$ that satisfies
\begin{equation}\label{eq:distinugishing}
 \E \cD(X,Z) - \E\cD(Y,Z) \geqslant \epsilon \quad \forall (Y,Z):\ \avminentr(Y|Z) \geqslant k.
\end{equation}
We will show how to construct an efficient algorithm that given $Z$ uses $\cD$ to predict $X$ with probability at least $2^{-k}$, contradicting (\ref{e:imp0}). The core of the algorithm is the procedure \ref{alg:predictor} described below. 

\begin{function}
\SetKwInOut{Input}{Input}\SetKwInOut{Output}{Output}
\SetKwFunction{BernoulliDistribution}{BernoulliDistribution}
\SetKwFunction{Predictor}{Predictor}
\DontPrintSemicolon
\Input{$z\leftarrow Z$, $[0,2]$-valued distinguisher $\cD'$}
\Output{$x\in\{0,1\}^{n}$}  
$b \gets 1$, $i\gets 1$ \;
\While{$b\not = 0$ {\bf and} $i < \ell$}{
$x\gets \left\{0,1\right\}^n$  \;
$b \gets $ \BernoulliDistribution{$\cD'(x,z)/2$}  \tcc*{outputs $1$ w.p. $\cD'(x,z)/2$} 
\eIf{$b=0$}{$i\gets i+1$ }{ \Return{x} }} 
\Return{$\bot$} \;
\caption{ Predictor($z,\cD',\ell$) }\label{alg:predictor} 
\end{function} 

$\mathrm{Predictor}(Z,\cD,\ell)$ samples an element $x\in \nstrings$ according to some probability distribution. This distribution captures the following intuition: as the advantage $\E \cD(X,Z) - \E\cD(Y,Z)$ is positive (as assumed in \eqref{eq:distinugishing}), we know that $x$ being the correct guess for $X$ is positively correlated with the value $\cD(x,Z)$. 
The probability that $\mathrm{Predictor}(Z,\cD,\ell)$ returns some 
particular value $x$ as guess for $X$ will be linear in $\cD(x,Z)$. 

$\mathrm{Predictor}(Z,\cD,\ell)$ may also output $\bot$, which means it failed to 
sample an $x$ according to this distribution. The probability of outputting $\bot$ goes 
exponentially fast to $0$ as $\ell$ grows.
\paragraph{A toy example: predicting $X$ when $Z$ is empty and $\cD$ is boolean.}
Suppose that  $\E \cD(X) - \E\cD(Y) \geqslant \epsilon$ for all $Y$ such that $\minentr(Y) \geqslant k$. And assume that $\cD(.)$ is boolean (not real valued as in our theorem). 
Then $\mathrm{Predictor}(\emptyset,\cD,\ell)$ will output a 
guess for $X$ that (if it's not $\bot$) is a random value $x$ satisfying $\cD(x)=1$. The probability that this guess for $X$ is correct equals $\E\cD(X)/|D|$ where $|D| = \sum_{x} \cD(x)$. Consider now the distribution $Y$ of min-entropy $k$ that maximizes $\E\cD(Y)$. 
We can assume that $Y$ is flat and supported on those $2^{k}$ elements $x$ for which the value $\cD(x)$ is the biggest possible. Observe that since $\E \cD(X) - \E\cD(Y) > 0$, we have $\E\cD(Y) < 1$ and since $\cD$ is boolean, the support of $Y$ contains all the elements $x$ satisfying $\cD(x)=1$. Therefore we obtain $\E\cD(Y) = 2^{-k}|D|$. Now we can estimate the predicting probability from below as follows:
\begin{align*}
 \Pr[ \textrm{$X$ is predicted  correctly} ] = \frac{\E\cD(X)}{|D|}  \geqslant \frac{\E\cD(Y)+\epsilon}{|D|} = 2^{-k} + \frac{\epsilon}{|D|}
\end{align*}
The above probability holds for $\ell=\infty$, i.e., when predictor never outputs $\bot$. 
For efficiency reasons, we must use a finite, and not too big $\ell$. The predictor will output $\bot$ with probability 
$\left(1-2^{-n}|D|\right)^{\ell}$ and thus
\begin{align*}
 \Pr[\textrm{we predcit $X$ in time $\mathcal{O}(\ell\cdot\mathrm{time}(D))$}] = 
 \left( 2^{-k} + \frac{\epsilon}{|D|} \right)\left(1-\left(1-2^{-n}|D|\right)^{\ell}\right)
\end{align*}
With a little bit of effort one can prove that setting $\ell = 1+ 2^{n-k}/\epsilon\approx 2^{\Delta}/\epsilon$ yields the success probability $2^{-k}$ independently of $|D|$.
\paragraph{Proof in general case - important issues}
Unfortunately, what we have proven above cannot be generalized easily to the 
case considered in the theorem, there are two obstacles. 
First, in the theorem we consider a conditional distribution $X|Z$ (i.e., the conditional part $Z$ is not empty as above).  
Unfortunately we cannot simply 
make the above argument separately for all possible choices 
$Z=z$ of the conditional part, as we  cannot guarantee that the conditional advantages $\epsilon(z)=\E\cD(X|Z=z,z)-\E\cD(Y|Z=z,z)$ are \emph{all} positive; we only know that their average $\epsilon = \E_{z\leftarrow Z}\epsilon(z)$ is positive. Second, 
so far we assumed that $\cD$ is boolean. This would only prove the theorem 
where the derived entropy in (\ref{e:important}) is against deterministic \emph{boolean} 
distinguishers, and this is not enough to conclude that we have the same amount of HILL entropy as discussed in Section~\ref{S:MH}. 
\paragraph{Actual proof - preliminaries}
\noindent 
For real-valued distinguishers in the conditional case, just invoking $\textsc{Predictor}(Z,\algD,\ell)$ on a $\algD$ satisfying (\ref{eq:distinugishing}), will not give a 
predictor for $X$ with advantage $>2^{-k}$ in general. Instead, we first have to transform $\algD$ into a new distingusiher $\algD'$ that has the same distinguishing advantage, and
for which we can prove that the predictor will work.

The way in which we modify $\cD$ depends on the distribution $Y|Z$ that minimizes the left-hand side of \eqref{eq:distinugishing}. This distribution can be characterized as follows:
\begin{lemma}[\cite{Skorski15}]\label{thm:characterizing_maximizer}
\noindent Given $\cD:\nstrings\times\mstrings \rightarrow [0,1]$ consider the following optimization problem
\begin{equation}\label{eq:maximizing_expectation}
\begin{array}{rl}
 \max\limits_{Y|Z} & \mathbb{E} \cD(Y,Z) \\
 \textup{s.t.} & \avminentr(Y|Z) \geqslant k
\end{array}
\end{equation}
The distribution $Y|Z=Y^{*}|Z$ satisfying $\avminentr(Y^{*}|Z) = k$ is optimal for  (\ref{eq:maximizing_expectation}) if and only if there exist real numbers $t(z)$ and a number $\lambda \geqslant 0$ such that for every $z$
\begin{enumerate}[(a)]
\item $\sum_{x} \max(\cD(x,z) - t(z),0) = \lambda$ 
\item If $0 < \mathbf{P}_{Y^{*}|Z=z}(x) < \max_{x'}\mathbf{P}_{Y^{*}|Z=z}(x')$ then $\cD(x,z) =t(z)$.
\item If $\mathbf{P}_{Y^{*}|Z=z}(x) = 0$ then $\cD(x,z) \leqslant t(z)$ 
\item If $\mathbf{P}_{Y^{*}|Z=z}(x) = \max_{x'}\mathbf{P}_{Y^{*}|Z=z}(x')$ then $\cD(x,z) \geqslant t(z)$
\end{enumerate}
\end{lemma}
\begin{proof}
The proof is a straightforward application of the Kuhn-Tucker conditions given in Appendix. 
\qed\end{proof}
\begin{remark}
The characterization can be illustrated in an easy and elegant way. First, it says that the area under the graph of $\cD(x,z)$ and above the threshold $t(z)$ is the same, no matter what $z$ is (see Figure \ref{fig:1}).

\begin{figure}[!ht]
\begin{minipage}{0.5\linewidth}
\centering
\includegraphics[scale=1]{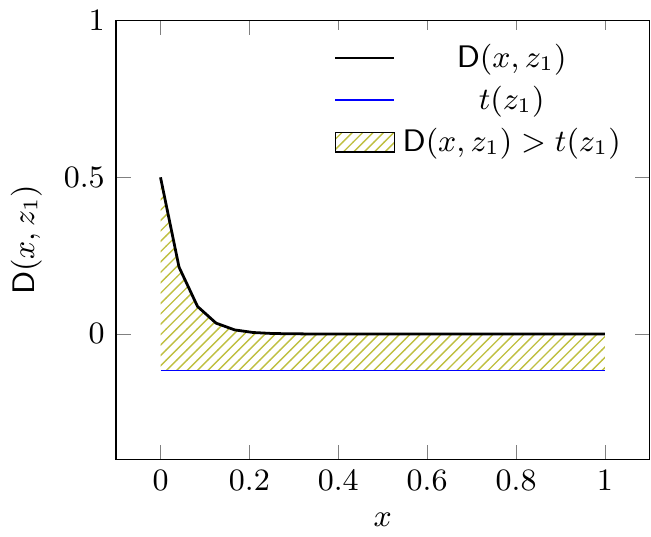} 
\end{minipage}
\begin{minipage}{0.5\linewidth}
\centering
\includegraphics[scale=1]{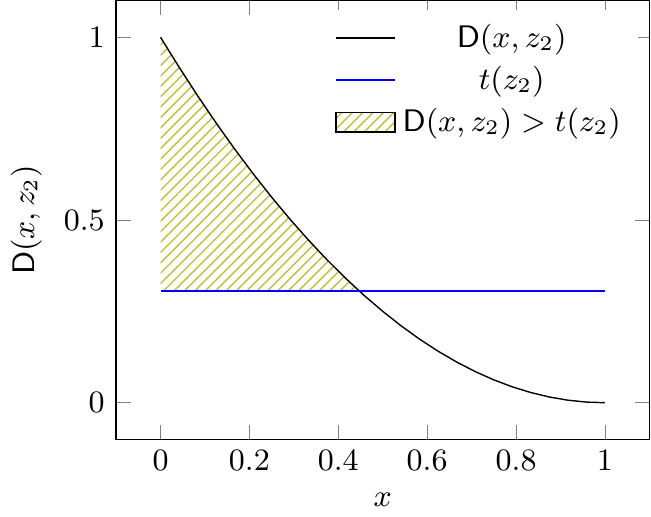} 
\end{minipage}
\caption{For every $z$, the (green) area under $\cD(\cdot,z)$ and above $t(z)$ equals $\lambda$}
\label{fig:1}
\end{figure}

\noindent
Second, for every $z$ the distribution $Y^{*}|Z=z$ is flat over the set $\left\{x:\ \cD(x,z) > t(z)\right\}$ and vanishes for $x$ satisfying $\cD(x,z) < t(z)$, see Fig. \ref{fig:2}.
\begin{figure}[!ht]
\centering
\includegraphics[scale=1]{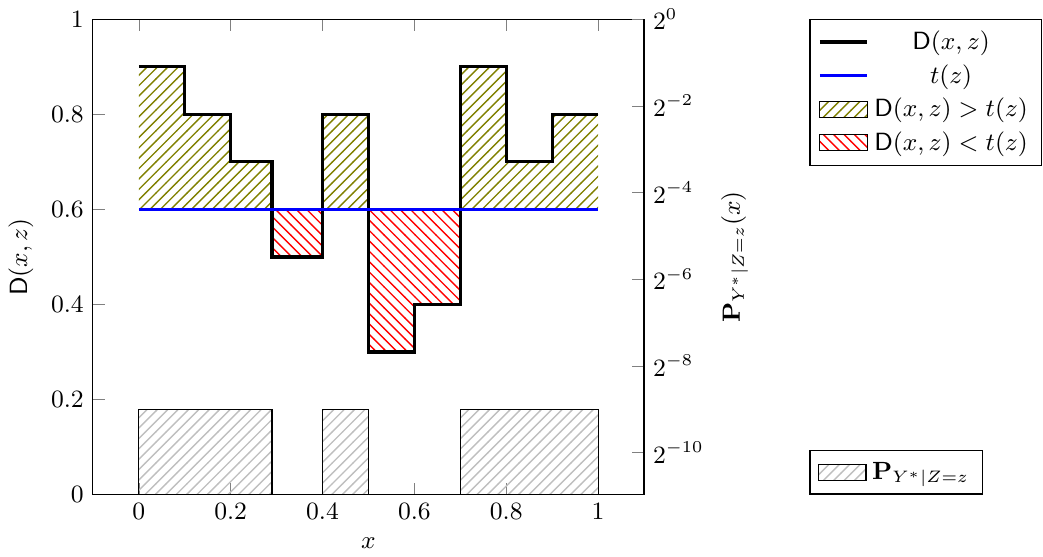} 
\caption{Relation between distinguisher $\cD(x,z)$, threshold $t(z)$ and distribution $Y^{*}|Z=z$.}
\label{fig:2}
\end{figure}
\end{remark}

\noindent
Note that because of ``freedom" in defining the distribution on elements $x$ satisfying $ \cD(x,z)=t(z)$ (\ref{thm:characterizing_maximizer}, point (b)), there could be many distributions $Y^{*}|Z$ corresponding to fixed numbers $\lambda$ and $t(z)$ that satisfy the characterization above, and this way are optimal to \eqref{eq:maximizing_expectation} with $k = \avminentr(Y^{*}|Z)$. For the sake of completeness we characterize bellow the all possible values of $k$ that match to $\lambda$ and $t(z)$. We note that this fact might be used to modify our nonuniform guessing algorithm into a uniform one.
\begin{corollary}\label{corollary:entropy_range}
Let $\cD:\nstrings\times\mstrings\rightarrow [0,1]$ and $\lambda \in (0,1)$.  Let $t(z)=t(\lambda,z)$ be the unique numbers that satisfy the condition (a) in Lemma \ref{thm:characterizing_maximizer}. Define
\begin{align}\label{eq:entropy_depending_on_lambda}
 k(\lambda) & = n -\log\left( \E_{z\leftarrow Z}\left[ 1/\mathbf{P}(\cD(U,z)\geqslant t(z))\right] \right), 
\end{align}
which is a non-decreasing right continuous function of $\lambda$. Let $k^{-}(\lambda) = \lim_{\lambda'\rightarrow \lambda^{-}}k(\lambda')$ and $k^{+}(\lambda)=\lim_{\lambda'\rightarrow \lambda^{+}}k(\lambda')=k(\lambda)$ be the one-sided limits.
Then for every $Y^{*}|Z$ of min-entropy $k=\avminentr(Y^{*}|Z)$ fulfilling (b),(c) and (d) we have $k^{-}\leqslant k\leqslant k^{+}$. Conversely, if $k$ satisfies $k^{-}\leqslant k\leqslant k^{+}$ then there exists a distribution $Y^{*}|Z$ fulfilling (b),(c) and (d) such that $\avminentr(Y^{*}|Z)=k$. 
\end{corollary}

\paragraph{Predicting given the thresholds $t(z)$.}
\noindent We use the numbers $t(z)$ to modify $\cD$ and  then we call the procedure\ref{alg:predictor} on the modified distinguisher. Lemma \ref{thm:unefficient_predictor} below shows that we could efficiently predict $X$ from $Z$, assuming we knew the numbers $t(z)$ for all $z$ in the support of $Z$ (later, we'll show how to efficiently approximate them)
\begin{lemma}\label{thm:unefficient_predictor}
Let $Y^{*}|Z$ be the distribution satisfying $\avminentr(Y^{*}|Z) = k$ and maximizing $\E\cD(Y,Z)$ over $\avminentr(Y|Z)\geqslant k$, where $k < n$ and $\cD$ satisfies \eqref{eq:distinugishing}. Let $t(z)$ be as in Lemma \ref{thm:characterizing_maximizer}.
Define 
\begin{equation}\label{eq:modified_distinguisher}
\cD'(x,z) = \max(\cD(x,z)-t(z),0)
\end{equation}
and set $\ell = 2\cdot 2^{n-k}\epsilon^{-1}$ in the algorithm \textsc{\ref{alg:predictor}}. Then we have
\begin{equation}\label{eq:predictor_succeeds}
 \Pr\left( \textsc{Predictor}(Z,\cD',\ell) = X\right) \geqslant 2^{-k}\left(1+2^{k-n}\epsilon\right)
\end{equation}
\end{lemma} 
\begin{proof}
\noindent We start by calculating  the probability on the left-hand side of\eqref{eq:predictor_succeeds}
\begin{Claim}\label{claim:guessingprobability}
For any\footnote{We will only use the claim for the distinguisher $\cD'$ as constructed above, but the claim holds in general.} $\cD'$, the algorithm \textsc{\ref{alg:predictor}} outputs $X$ given $Z=z$ with probability
\begin{align}\label{eq:guessingprobability}
\Pr_{X,Z}\left( \left.\textsc{Predictor}(Z,\cD',\ell) =X \right| Z=z \right) & = 2^{-n-1} g\left(\frac{\E\cD'(U,z)}{2}\right)\cdot \E\cD'(X|Z=z,z)
\end{align}
where $U$ is uniform over $\{0,1\}^{n}$ and $g$ is defined by $g(d) = \frac{1-(1-d)^{\ell}}{d}$ (so $g(d)\approx 1/d$ for large $\ell$)
\end{Claim}\label{claim:worstdistribution}
\begin{proof}[of Claim]
It is easy to observe that
\begin{align}
\Pr[\left.\textsc{Predictor}(z,\cD',\ell)= x \right|\textsc{Predictor}(z,\cD',\ell)\not=\bot] & =\frac{\cD'(x,z)}{\sum\limits_{x} \cD'(x,z)}
\end{align}
In turn, for every round $i=1,\ldots,\ell$ of the execution, the probability that \textsc{Predictor} stops and outputs $x'$ is equal to $\Pr[U=x']\cD'(x',z)/2=2^{-n-1}\cD'(x',z)$, the probability 
that it outputs anything (and thus leaves the while loop) is thus
$\sum_{x'} \Pr[U=x']\cdot \left(1-\frac{\cD'(x',z)}{2}\right)=1-\frac{\E\cD'(U,z)}{2}$. So the probability of not leaving the while loop for $\ell$ rounds (in this case the output is $\bot$) is
\begin{align}
 \Pr[\textsc{Predictor}(z,\cD',\ell)= \bot] = 1-\left(1-\frac{\E\cD'(U,z)}{2}\right)^{\ell}
\end{align}
Combining the last two formulas we obtain
\begin{align}
\Pr[\textsc{Predictor}(z,\cD')= x] & = 2^{-n-1} g( \E\cD'(U,z)/2)\cdot \cD'(x,z) 
\end{align}
Hence
\begin{align}
 \Pr[\textsc{Predictor}(z,\cD') = X|Z=z] & = \sum\limits_{x}\Pr[\textsc{Predictor}(z,\cD')= x,X=x|Z=z] \nonumber \\
 & = \sum\limits_{x}\Pr[\textsc{Predictor}(z,\cD')= x]\Pr[X=x|Z=z] \nonumber \\
 & = 2^{-n-1} g(\E\cD'(U,z)/2) \sum\limits_{x} \cD'(x,z)\Pr[X=x|Z=z] \nonumber \\
 & = 2^{-n-1} g(\E\cD'(U,z)/2) \E\cD'(X|Z=z,z)
\end{align}
and the claim follows.
\qed\end{proof}
Now we can see why we cannot apply the algorithm \textsc{\ref{alg:predictor}} using the distinguisher $\cD$ satisfying only \eqref{eq:distinugishing} directly. According to the last formula, the success probability would be an averaged sum of products $g(\E\cD(U,z))\cdot \E\cD(X|Z=z,z)$ over $z$. We know the average of the second factors of these products, but in general cannot compare the values of $\E\cD(U,z)$ for different $z$'s. The crucial observation is that the distinguisher $\cD'$ we defined satisfies the same inequality \eqref{eq:distinugishing} as $\cD$ (though, $\cD'$ has the range $[0,2]$ not $[0,1]$ as $\cD$). Moreover $\cD'$ has a special form which allows us to simplify expression \eqref{eq:predictor_succeeds}. The details are given in the next two claims
\begin{Claim}\label{claim:replacing_distinguisher}
We have $\E \cD'(X,Z) - \E\cD'(Y,Z) \geqslant \epsilon$ for all $Y|Z:\ \avminentr(Y|Z)\geqslant k$
\end{Claim}
\begin{proof}[of Claim]
We argue that (a): $\E \cD'(X,Z) - \E\cD'(Y^{*},Z)  \geqslant \E \cD(X,Z) - \E\cD(Y^{*},Z)$
and (b): $Y^{*}|Z$ maximizes $\cD'(Y,Z)$ over $\avminentr(Y|Z)\geqslant k$.  For the proof of (a), observe that by \eqref{eq:modified_distinguisher} we have $\cD'(x,z) \geqslant \cD(x,z) - t(z)$ for every $x$ and $z$. Hence $\E\cD'(X,Z) \geqslant \E\cD(X,Z) - t(z)$. Moreover, if $\cD(x,z)-t(z) < 0$ then Lemma \ref{thm:characterizing_maximizer} implies $\mathbf{P}_{Y^{*}|Z=z}(x) = 0$ and thus $\E \cD'(Y^{*}|Z=z,z) = \E\cD(Y^{*}|Z=z)-t(z)$. Hence, for all $z$ we have
\begin{equation*}
\E\cD'(X|Z=z) - \E \cD'(Y^{*}|Z=z,z) \geqslant \E\cD(X|Z=z,z)-\E\cD(Y^{*}|Z=z,z)
\end{equation*}
The proof of (a) follows now by taking the average over $z$. The proof of (b) follows by observing that $\cD'$ satisfies the  characterization in \eqref{thm:characterizing_maximizer} with $t(z) = 0$ for all $z$.
\qed\end{proof}

\begin{Claim}\label{claim:uniform_expectations_are_equal}
The exists a number $\lambda' \in (0,1)$ such that $\E\cD'(U,z) = \lambda'$ for every $z$. 
\end{Claim}
\begin{proof}
Lemma \ref{thm:characterizing_maximizer} implies $\sum_x \cD'(x,z) = \lambda$ for every $z$. We can define $\lambda' = 2^{-n}\lambda$ and then it remains to show $\lambda < 2^n$ and $\lambda >0$. Observe that the case
$t(z) < 0$ in Lemma \ref{thm:characterizing_maximizer} is possible if and only if $\mathbf{P}_{Y^{*}|Z=z}(x) = \max_{x'}\mathbf{P}_{Y^{*}|Z=z}(x')$ for all $x$, which means $\minentr(Y^{*}|Z=z) = n$. Since $k < n$, we have $t(z) \geqslant 0$ for at least one $z$ and then $\lambda = \sum_{x}\max(\cD(x,z)-t(z),0) \leqslant \sum_{x}\cD(x,z)$ which essentially means $\lambda \leqslant 2^{n}$. Lemma \ref{thm:characterizing_maximizer} guarantees that $\lambda\geqslant 0$ , therefore we need to show that $\lambda\not \in \{0,2^n\}$. Observe that if $\lambda = 0$ then the condition $\sum_x \cD'(x,z) = \lambda$ implies $\cD'(x,z)=0$ for all $x$ and $z$, contradicting to Claim \ref{claim:replacing_distinguisher} because $\epsilon > 0$. In turn, if $\lambda = 2^{n}$ then from Lemma \ref{thm:characterizing_maximizer} we get $\cD(\cdot,z) \equiv 1$ and $t(z) = 0$ for all $z$ such that $t(z)\geqslant 0$.  This is possible only if $\mathbf{P}_{Y^{*}|Z=z}(x) = \max_{x'}\mathbf{P}_{Y^{*}|Z=z}(x') $ for all $x$ which means $\minentr(Y^{*}|Z=z) = n$ if $t(z)\geqslant 0$. But then $\minentr(Y^{*}|Z=z) = n$ for all $z$ which contradicts $k<n$.
\qed\end{proof}
\noindent To calculate the success probability we need one more observation. The following claim shows that support of $\cD'$ is contained in the support of $Y^{*}$.
\begin{Claim}\label{claim:conditional_expectations}
For every $z$ we have
\begin{align}\label{eq:conditional_expectations}
   \E \cD'(Y^{*}|Z=z,z) =  \E\cD'(U,z) \cdot 2^{n}\max_{x'}\mathbf{P}_{Y^{*}|Z=z}(x').
\end{align}
\end{Claim}
\begin{proof}[of Claim]
By Lemma \ref{thm:characterizing_maximizer}, $\cD(x,z) > t(z)$ only if $\mathbf{P}_{Y^{*}|Z=z}(x) = \max_{x'}\mathbf{P}_{Y^{*}|Z=z}(x')$ therefore 
\begin{align*}
 \E \cD'(Y^{*}|Z=z,z) &=   \sum\limits_{x} \max(\cD(x,z)-t(z),0)\mathbf{P}_{Y^{*}|Z=z}(x) \nonumber \\
 & = \sum\limits_{x} \max(\cD(x,z)-t(z),0) \max_{x'}\mathbf{P}_{Y^{*}|Z=z}(x'),
\end{align*}
and the claim follows by the definition of $\cD'$.
\qed\end{proof}
\noindent Now we are ready to prove the main result. From Claim \ref{claim:guessingprobability}  and Claim \ref{claim:uniform_expectations_are_equal} we obtain
\begin{align}\label{eq:predicting_probability_estimate_1}
 \Pr\left( \textsc{Predictor}(Z,\cD',\ell) =X \right) & =  2^{-n-1} \E_{z\leftarrow Z}\left[ g(\lambda'/2)\cdot\cD'(X|Z=z,z)\right] \nonumber \\ 
 & = 2^{-n-1}g(\lambda'/2) \cdot \E\cD'(X,Z) 
\end{align}
Claim \ref{claim:replacing_distinguisher} applied to $Y=Y^{*}$ yields now the following estimate
\begin{align}\label{eq:predicting_probability_estimate_2}
 \Pr\left( \textsc{Predictor}(Z,\cD',\ell) =X \right) & \geqslant 2^{-n-1}g(\lambda'/2) \cdot\left( \E\cD'(Y^{*},Z)+\epsilon \right).
\end{align}
Observe that Claim \ref{claim:conditional_expectations}, Claim \ref{claim:uniform_expectations_are_equal}, and $\avminentr(Y^{*}|Z)=k$ imply
\begin{align}
 \E\cD'(Y^{*},Z) & =   \E_{z\leftarrow Z}\left[ \cD'(Y^{*}|Z=z,z) \right] 
= \E_{z\leftarrow Z}\left[ \E\cD'(U,z) \cdot 2^{n}\max_{x'}\mathbf{P}_{Y^{*}|Z=z}(x') \right]\nonumber \\
& 
= 2^{n}\lambda'\cdot  \E_{z\leftarrow Z}\left[ \max_{x'}\mathbf{P}_{Y^{*}|Z=z}(x') \right]
 =  2^{n-k}\lambda'
\end{align}
Plugging this into \eqref{eq:predicting_probability_estimate_2} we get the following bound
\begin{align}\label{eq:predicting_probability_estimate}
 \Pr\left( \textsc{Predictor}(Z,\cD',\ell)  = X \right) & \geqslant 2^{-n-1}g(\lambda'/2) 
\cdot\left( 2^{n-k}\lambda'+\epsilon \right)  \nonumber \\
& = 2^{-k}\left(1-(1-\lambda'/2)^{\ell} \right)\left(1+ \frac{2^{k-n-1}\epsilon}{\lambda'/2}\right)
\end{align}
To give a lower bound on the success probability it remains to minimize the last expression over  $\lambda' \in (0,1)$. This is answered below
\begin{Claim}\label{claim:probability_expression_analysis}
Let $h(s) =(1-(1-s)^{\ell})(1 + a s^{-1})$, where $a  >0$ and $\ell \geqslant 1+a^{-1}$. Then $h(s) \geqslant h(1) = 1+a$ for all $s\in [0,1]$.
\end{Claim}
\begin{proof}[of Claim]
The proof uses standard calculus and is given in the appendix.
\qed\end{proof}

\paragraph{Computing $t(z)$ from $\lambda$}
So far, we have shown how to construct the predicting algorithm provided that we are given the numbers $t(z)$. Now we will prove that one can compute them \emph{approximately} and use \emph{successfully} in place of the original ones. We start with a few useful facts about the auxiliary function $g$ already introduced in Claim \ref{claim:guessingprobability} in the proof of Lemma \ref{thm:unefficient_predictor}. Below we summarize its fundamental properties.
\begin{lemma}\label{thm:main_function_properties}
For $\ell > 1$ the function $g(d) = \frac{1-(1-d)^{\ell}}{d}$ on $[0,1]$ satisfies:
\begin{enumerate}[(a)]
\item $g$ is continuous at $0$ and decreasing
\item $g$ is convex
\item for any $d_2 > d_1$ we have  $ g(d_2) >  g(d_1)\left(1- \frac{\ell}{2} \cdot |d_2-d_1| \right) $
\end{enumerate}
\end{lemma}
\begin{proof}[of Lemma]
The proof uses elementary calculus and is referred to the appendix
\qed\end{proof}
The entire solution is based on the next two lemmas. The first lemma is based on the intuition that replacing $\cD$ by a distinguisher which approximates it close enough should not affect the success probability of $\textsc{Predictor}(Z,\cD,\ell)$ very much. For technical reasons we present this statement assuming \emph{one-sided $\mathcal{L}^1$-approximation}.  
The second lemma describes an efficient algorithm which obtains $\lambda$ as a \emph{ hint} on its input and computes approximations for $t(z)$ from below, for every $z$. 
\begin{lemma}\label{lemma:predicting_using_approximated_distinugisher}
Let $\cD_1,\cD_2 : \nstrings\times\mstrings \rightarrow [0,1]$ be any two functions satisfying
\begin{enumerate}[(a)]
\item $\cD_2(x,z)\geqslant \cD_1(x,z) $ for all $x,z$
\item $\E\cD_2(U,z) - \E\cD_1(U,z)  \leqslant \delta$ for all $z$
\end{enumerate}
Then we have
\begin{equation}
\Pr\left( \textsc{Predictor}(Z,\cD_2,\ell) =X \right) \geqslant (1-\ell\delta/2)\Pr\left( \textsc{Predictor}(Z,\cD_1,\ell) =X \right) 
\end{equation}
\end{lemma}
\begin{proof}[of Lemma]
We have
\begin{align}
 \Pr\left( \textsc{Predictor}(z,\cD_2,\ell) =X|Z=z \right) & = g(\E\cD_2(U,z)) \E\cD_2(X|Z=z,z) \nonumber \\ 
 & \geqslant g(\E\cD_2(U,z)) \E\cD_1(X|Z=z,z),
\end{align}
where the inequality follows from $\cD_2 \geqslant \cD_1 \geqslant 0$. The assumptions (a) and (b) imply $\left| \E\cD_1(U,z) - \E\cD_2(U,z)\right| \leqslant \delta$ for every $z$. From property (c) in Lemma \ref{thm:main_function_properties} it follows that 
\begin{align*}
 g(\E\cD_2(U,z)) \geqslant g(\E\cD_1(U,z)) (1-\ell\delta/2)
\end{align*}
for every $z$. Combining the last two estimates we get
\begin{align}
\Pr\left( \textsc{Predictor}(z,\cD_2,\ell) =X|Z=z \right)  & \geqslant (1-\ell\delta/2)\cdot g(\E\cD_1(U,z)) \E\cD_1(X|Z=z,z) \nonumber \\
& = (1-\ell\delta/2)\cdot \Pr\left( \textsc{Predictor}(z,\cD_1) =X|Z=z \right)
\end{align}
Taking the average over $z\leftarrow Z$ completes the proof.
\qed\end{proof}
\begin{lemma}\label{lemma:find_threshold}
Let $\cD :\nstrings \rightarrow [0,1]$ be any function computable in time $s$, let $\lambda \in (0,1)$ and $t \in [0,1]$ be a number such that $\E\max(\cD(U)-t,0) = \lambda$. There exists a probabilistic algorithm $\textsc{FindThreshold}(\cD,\lambda,\delta,N)$ that runs in time $\mathcal{O}\left(\log(1/\delta) N \cdot\mathrm{time}(\cD)\right)$ and with probability at least $1-2\log(12/\delta)\mathrm{e}^{-N\delta^2/3}$ outputs a number $t'$ such that $\E\max(\cD(U)-t',0)  \in [\lambda, \lambda+\delta]$. In particular, $t'\leqslant t$.
\end{lemma}
\begin{proof}[of Lemma]
The idea is pretty simple: given $t'$ we approximate values $\E\max(\cD(U)-t',0)$ by sampling and by comparing the result with $\lambda$, we can find the right value of $t'$ using binary search. This corresponds to finding a blue line on Fig. \ref{fig:2} such that the green area above is sufficiently close to $\lambda$.
\begin{function}[h]
\SetKwInOut{Input}{Input}\SetKwInOut{Output}{Output}
\SetKwFunction{BernoulliDistribution}{BernoulliDistribution}
\SetKwFunction{FindThreshold}{FindThreshold}
\DontPrintSemicolon
\Input{$\cD: \nstrings\rightarrow [0,1]$, $\lambda \in (0,1)$, parameters $\delta,N$}
\Output{$t'$ such that $ \E\max(\cD(U)-t',0) \in [\lambda,\lambda+\delta]$}  
$t^{-}\gets -1, t^{+}\gets 1$ \;
\Repeat{ $t^{+} - t^{-} \leqslant \frac{\delta}{12}$} {
$t' \gets (t^{-}+t^{+})/2$ \;
$x_1,\ldots,x_{N} \leftarrow U$ \tcc*{fresh values every time}
$\lambda' \gets N^{-1}\sum_{j=1}^{N} \max\left(\cD(x_j) - t',0 \right)$ \tcc*{ $\lambda' \approx \E\max(\cD(U)-t_i,0)$}
\uIf{$\lambda' > \lambda+\frac{2\delta}{3}$}{
$t^{-}\gets t'$ \;
}
\uElseIf{$\lambda' < \lambda + \frac{\delta}{3}$}{
$t^{+} \gets t'$ \;
}
\Else{
 \Return{$t'$} \;
}
}
\uIf{$t'<-1+ \frac{\delta}{12}$}{
$t'\gets -1$\;
}
\Return{$t'$}\;
\caption{ FindThreshold($\cD,\lambda,\delta,N$) }\label{alg:find_threshold} 
\end{function}
 
\noindent
The function $h(t')= \E\max(\cD(U)-t',0)$ is clearly non-increasing with respect to $t'$ and changes from $1+\E\cD(U)$ at $t'=-1$ to $0$ for $t=1$. Moreover, it is strictly decreasing in a small neighborhood of $t'=t$ and for all $t' < t$. Indeed, since $\lambda > 0$ there is at least one $x$ such that $\cD(x) > t$. Taking $t' < t''\leqslant\min_{x: \cD(x) > t} \cD(x)$ we see that $h(t')-h(t'')\geqslant 2^{-n}(t''-t') >0$. Hence, $t' >t$ implies $\E\max(\cD(U)-t',0) < \E\max(\cD(U)-t,0)=\lambda$. This proves the second part of the statement. 
Denote by $\lambda'_i, t'_i, t^{-}_i,t^{+}_i$ the values assigned in round $i$ to $\lambda',t',t^{-},t^{+}$ respectively. Observe that by the Chernoff Bound\footnote{We use the following version: let $X_1,\ldots,X_N$ be $[0,1]$-valued independent random variables, let $X = \sum_{i=1}^{N}X_i$ and $\mu = \E X$. Then $\Pr\left( \left| X-\mu\right| > \delta\mu \right) < 2\exp(-\mu \delta^2/3)$ } and the union bound over at most $\log(12/\delta)$ rounds of the execution, with probability $p=1-2\log(12/\delta)\exp(-N\delta^2/3)$ we have $\left| \lambda'_i -h(t_i)\right| < \frac{\delta}{12}$ for every round $i$. Note that with the same probability the algorithm satisfies the invariant property: if there is $t_0 \in [t^{-}_i,t^{+}_i]$ such that $h(t_0) \in \left[\lambda+\frac{5\delta}{12},\lambda+\frac{7\delta}{12}\right]$ and the algorithm jumps to round $i+1$ then $t_0 \in \left[t^{-}_{i+1},t^{+}_{i+1}\right]$. Suppose that $h(t_0) \in \left[\lambda+\frac{5\delta}{12},\lambda+\frac{7\delta}{12}\right]$ for some $t_0 \in [-1,1]$. Now we have two possibilities:
either we terminate with $t_i$ such that $\lambda_i \in \left[\lambda+\frac{\delta}{3},  \lambda + \frac{2\delta}{3}\right]$ which means $h(t_i) \in \left[\lambda+\frac{3\delta}{12},\lambda+\frac{7\delta}{12}\right] $ and we are done, or we will eventually find such $t_0$ up to an error $\frac{\delta}{12}$. Since $|h(t_2)-h(t_1)| \leqslant |t_2-t_1|$ for any $t_1,t_2$, the returned number $t'$ satisfies $h(t_0) - \frac{\delta}{12} \leqslant h(t') \leqslant  h(t_0) + \frac{\delta}{12}$, in particular it satisfies the desired inequality. It remains to consider the case when either $h(t) < \lambda+\frac{5\delta}{12}$ for all $t$ or $h(t) > \lambda + \frac{7\delta}{12}$. Since $h(1)=0$ the second is clearly impossible. In the first case we have $h(t) \leqslant h(-1) < \lambda+\frac{5\delta}{12}$, which means that in every round $i$ we have $t^{-}_i=-1$ and either we terminate with $t_i$ such that $\lambda'_i \in \left[\lambda+\frac{\delta}{3},  \lambda + \frac{2\delta}{3}\right]$ which means $h(t_i) \in \left[\lambda+\frac{3\delta}{12},\lambda+\frac{7\delta}{12}\right] $ and we are done, or 
in every round $i$ we do the assignment $t^{+}_{i+1}=t_i$ which yields $t_i =-1+ 2^{-i+1}$ and the main loop halts with $t_i < -1+\frac{\delta}{12}$. The algorithm outputs then $-1$ which satisfies the desired inequality, because of the assumption $h(-1) < \lambda+\frac{5\delta}{12}$ and the trivial inequality $h(-1)\geqslant 1\geqslant \lambda$.
\qed\end{proof}
\noindent Let $\cD'$ be as in Lemma \ref{thm:unefficient_predictor}. Let $t'(z) = \textsc{FindThreshold}(\cD,\lambda,\delta,N)$, define $D''(x,z) = \max(\cD(U,z)-t'(z),0)$. Denote by $\Pr[bad]$ the probability that $\E D''(U,z) \not\in [\lambda,\lambda+\delta]$  (i.e. probability of failure of the algorithm \ref{alg:find_threshold}). If the event $bad$ doesn't occur then $\cD'' \geqslant \cD'$ and $\E\cD''(U,z) \leqslant \E\cD'(U,z)+\delta$. Applying the last two claims we obtain
\begin{align}
\Pr\left[ \textsc{Predictor}(z,\cD'',\ell) \right] \geqslant 2^{-k}\left(1+2^{k-n}\epsilon\right)\cdot \left(1-\frac{\ell \delta}{2}\right)\Pr[\neg bad]
\end{align}
By the elementary inequality $(1+s)(1-s/4)^2\geqslant 1$ valid for $s\in [0,1]$, for this probability to be bigger than $2^{-k}$ it is enough to require 
\begin{align}
 \ell \delta / 2 & \leqslant  2^{k-n}\epsilon / 4 \\
2\log(12/\delta)\exp(-N\delta^2)/3) & \leqslant  2^{k-n}\epsilon / 4
\end{align}
The solution for the first inequality is $\delta = \mathcal{O}(2^{2(k-n)}\epsilon^2)$ which implies $\delta \ll \epsilon$. The second one gives us $N =\Omega\left( (1/\delta)^2( \log \log(1/\delta) + n-k + \log(1/\epsilon) \right)$ which can be simplified to $N =\Omega\left( (1/\delta)^2( \log(1/\delta) \right)$. The total running time is (up to a constant factor) the time needed for invoking $\mathcal{O}\left(\ell\cdot N\log(1/\delta)\right)= \mathcal{O}\left( (2^{\Delta}/\epsilon)^5\log^2\left( 2^{\Delta}/\epsilon \right)\right)$ times of the distinguisher $\cD$ .

%% file: U2M2.tex
\section{Proof of Lemma \ref{thm:characterizing_maximizer}}
\begin{proof}
Consider the following linear optimization program
\begin{equation}\label{eq:maximizing_expectation_linearized}
 \begin{aligned}
 & \underset{P_{x,z},a_z}{\text{maximize}}  &  \sum\limits_{x,z} \cD(x,z)P(x,z) \\
 & \text{subject to} & -P_{x,z} \leqslant 0  &,\  (x,z) \in \nstrings\times\mstrings \\
 & &  \sum\limits_{x}  P_{x,z} - \mathbf{P}_Z(z) = 0 &,\  z\in\mstrings \\
 & & P_{x,z} -a_z \leqslant 0 &,\ z \in\mstrings \\
 && \sum\limits_{z} a_z - 2^{-k}  \leqslant 0 &
 \end{aligned}
\end{equation}
This problem is equivalent to (\ref{eq:maximizing_expectation}) if we define $\mathbf{P}_{Y,Z}(x,z)=P(x,z)$ and replace the condition $\sum_{z}\max_{x}\mathbf{P}_{Y,Z}(x,z)\leqslant 2^{-k}$, which is equivalent to $\avminentr(Y|Z)\geqslant k$, by the existence of numbers $a_z \geqslant \max_{x}\mathbf{P}_{Y,Z}(x,z)$ such that $\sum_{z}a_z\leqslant 2^{-k}$. The solutions of (\ref{eq:maximizing_expectation_linearized}) can be characterized as follows:
\begin{Claim}
The numbers $(P_{x,z})_{x,z},(a_z)_z$ are optimal for (\ref{eq:maximizing_expectation_linearized})
if and only if there exist numbers $\lambda^{1}(x,z)\geqslant 0$, $\lambda^{2}(z) \in \mathbb{R}$, $\lambda^{3}(x,z) \geqslant 0$, $\lambda^4 \geqslant 0$ such that
\begin{enumerate}[(a)]
\item ${\cD}(x,z) =  -\lambda^1(x,z) + \lambda^2(z) + \lambda^3(x,z) $ and $0 = -\sum_{x} \lambda^3(x,z) + \lambda^4$
\item We have $\lambda^1(x,z) = 0$ if $P_{x,z} > 0$, $\lambda^3(x,z) = 0$ if $P_{x,z} < a_z$, $\lambda^{4} = 0$ if $\sum_z a_z < 2^{-k}$.
\end{enumerate}
\end{Claim}
\begin{proof}[of Claim]
This is a straightforward application of KKT conditions.
\qed\end{proof}
\noindent It remains to apply and simplify the last characterization. Let $(P^{*}_{x,z})_{x,z},(a^{*}_z)_z$ be optimal for \eqref{eq:maximizing_expectation_linearized}, where $P^{*}(x,z)=\mathbf{P}_{Y^{*},Z}(x,z)$, and $\lambda^{1}(x,z),\lambda^{2}(z),\lambda^{3}(x,z),\lambda^{4}(x)$ be corresponding multipliers given by the last claim. Define $t(z) = \lambda^{2}(z)$ and $\lambda = \lambda^{4}$. Observe that for every $z$ we have $a^{*}_z \geqslant \max\limits_{x}\mathbf{P}(x,z) \geqslant 2^{-n}\mathbf{P}_Z(z) > 0$ and thus for every $(x,z)$ we have 
\begin{equation}\label{eq:55}
\lambda^{1}(x,z)\cdot \lambda^{3}(x,z)= 0
\end{equation}
If $P^{*}(x,z)=0$ then $P^{*}(x,z) < a^{*}(z)$ and $\lambda^{3}(x,z) = 0$, hence $\cD(x,z) \leqslant t(z) $ which proves (c). If $P^{*}(x,z) = \max_{x'}P^{*}(x,z)$ then $P^{*}(x,z) < 0$ and $\lambda^{1}(x,z) = 0$ which proves (d). 
Finally observe that \eqref{eq:55} implies
\begin{align*}
 \max(\cD(x,z)-t(z),0)& = \max(-\lambda^{1}(x,z)+\lambda^{3}(x,z),0)  = \lambda^{3}(x,z)
\end{align*}
Hence, the assumption $\sum_{x}\lambda^{3}(x,z) = \lambda^{4}=\lambda$ proves (a).

Suppose now that the characterization given in the Lemma is satisfied. Define $P^{*}(x,z) = \mathbf{P}_{Y,Z}(x,z)$ and $a_z = \max_{z}\mathbf{P}_{Y^{*},Z}(x,z)$, let $\lambda^{3}(x,z) = \max(\cD(x,z) - t(z),0)$, $\lambda^{1}(x,z)=\max(t(z)-\cD(x,z),0)$ and $\lambda^{4}=\lambda$. 
We will show that these numbers satisfy the conditions described in the last claim. By definition we have $-\lambda^{1}(x,z) + \lambda^{2}(z) + \lambda^{3}(x,z) = \cD(x,z)$, by the assumptions we get $\sum_{x} \lambda^{3}(x,z) = \lambda = \lambda^{4}$. This proves
part (a).\ Now we verify the conditions in (b). Note that $\cD(x,z) < t(z)$ is possible only if $\mathbf{P}_{Y^{*}|Z=z}(x)=0$ and $\cD(x,z) > t(z)$ is possible only if $\mathbf{P}_{Y^{*}|Z=z}(x) = \max_{x'}\mathbf{P}_{Y^{*}|Z=z}(x')$. Therefore, if $\mathbf{P}_{Y,Z}(x,z) > 0$ then we must have $\cD(x,z) \geqslant t(z)$ which means that $\lambda^{1}(x,z)=0$. Similarly if $\mathbf{P}_{Y,Z}(x,z) < \max_{z}\mathbf{P}_{Y^{*},Z}(x,z)$ then $\cD(x,z) \leqslant t(z)$ and $\lambda^{3}(x,z) =0$. Finally, since we assume $\avminentr(Y^{*}|Z) = k$ we have $\sum_{z} a_z = 2^{-k}$ and thus there is no additional restrictions on $\lambda^{4}$.
\qed\end{proof}

\section{Proof of Corollary\ref{corollary:entropy_range}}

\begin{proof}[of Corollary]
Let $y_{\max}(z) = \max_{x'}\mathbf{P}_{Y|Z=z}(x')$. Consider the function
\begin{align}
f^{\delta}_z(x) = \left\{
\begin{array}{rl}
 y_{\max}(z) + \delta, & \cD'(x,z) > t(z) \\
 \frac{1-\#\left\{x:\ \cD'(x.z)>t(z) \right\}\cdot (y_{\max}+\delta)}{\#\left\{x:\ \cD'(x,z)=t(z) \right\}}, & \cD'(x,z) = t(z) \\
 0, & \cD'(x,z) < t(z)
\end{array}
\right.
\end{align}
This function defines a distribution that satisfies 
\begin{align}
f^{\delta}_{z}(x) \leqslant \max_{x'}f^{\delta}_{z}(x') \quad \forall x: \cD'(x,z)\leqslant t(z)
\end{align}
if and only if $\delta$ satisfies
\begin{align}\label{eq:mass_shifting_condition}
 \frac{1}{ \#\left\{x:\ \cD'(x.z)\geqslant t(z) \right\} } \leqslant y_{\max}(z) + \delta \leqslant \frac{1}{ \#\left\{x:\ \cD'(x.z)>t(z) \right\} }
\end{align}
In particular these conditions are satisfied for $\delta = 0$. Suppose now that there are $z_i$ and $x_i$ for $i=1,2$ such that $
  0< \mathbf{P}_{Y^{*}|Z=z_i}(x_i) < \max\limits_{x'}\mathbf{P}_{Y^{*}|Z=z}(x')$. Define $\delta$ by
\begin{align*}
 \delta  = \min\left(  y_{\max}(z_1)-\frac{1}{ \#\left\{ x: \cD'(x,z_1)\geqslant t(z_1) \right\} },\ \frac{1}{ \#\left\{ x: \cD'(x,z_2)> t(z_2) \right\} } - y_{\max}(z_2) \right)
\end{align*}
By Lemma \ref{thm:characterizing_maximizer} we immediately obtain that $\delta \geqslant 0$. It follows easily from the definition of $\delta$ that the number $-\delta$ satisfies \eqref{eq:mass_shifting_condition} with $z=z_1$ and that $\delta$ satisfies \eqref{eq:mass_shifting_condition} for $z=z_2$. We can see now that if we replace the distribution $Y^{*}|Z=z_1$ by $f^{-\delta}_{z_1}$ and the distribution $Y^{*}|Z=z_2$ by $f^{\delta}_{z_2}$ then we obtain the distribution $Y'|Z$ satisfying conditions in Lemma \ref{thm:characterizing_maximizer} and $\avminentr(Y'|Z)= k$. Finally, observe that $\delta = \frac{1}{ \#\left\{ x: \cD'(x,z_2)> t(z_2) \right\} } - y_{\max}(z_2) $ means that the distribution $Y'|Z=z_2$ is uniform on $\left\{ x: \cD'(x.z_2)>t(z_2) \right\}$. In turn, if $\delta=y_{\max}(z_1)-\frac{1}{ \#\left\{ x: \cD'(x,z_1)\geqslant t(z_1) \right\} }$ then the distribution $Y'|Z=z_1$ is uniform on $\left\{x: \cD'(x,z_1)\geqslant t(z_1)\right\}$.
\qed\end{proof}

\section{Proof of Claim \ref{claim:probability_expression_analysis}, Lemma \ref{thm:unefficient_predictor}}
\begin{proof}
We check that $\lim_{s\to 0}h(s) = a\ell$ and thus the function $h$ is continuous on the interval $[0,1]$. This means that $h$ attains its minimum  at some point $s=s_0$.
There is nothing to prove if $s_0 \in \{0,1\}$. Suppose that $s_0 \in (0,1)$. Then we must have $\left.\frac{\partial h}{\partial s}\right|_{s=s_0} = 0$. The first derivative of the function $h$ is given by the following formula
\begin{align}
\frac{\partial h}{\partial s} &= \frac{s \ell (a+s) (1-s)^{\ell-1}+a \left((1-s)^\ell-1\right)}{s^2}\nonumber \\
& = \frac{-a+(1-s)^{\ell-1}\left(a(1-s)+(a+s)\ell s \right)}{s^2}
\end{align}
Therefore for $s=s_0$ we obtain $(1-s_0)^{\ell-1}  = \frac{a}{a(1-s_0)+(a+s_0)\ell s_0 }$ and hence
\begin{align}
 h(s_0) &= (1-(1-s_0)\cdot (1-s_0)^{\ell-1})\left(1+a s_0^{-1}\right) \nonumber \\
 & = \frac{(a+s_0)^2\ell }{a(1-s_0)+(a+s_0)\ell s_0}
\end{align}
Note that the last expression is increasing with respect to $\ell$ and that from the assumption we have $\ell > \frac{1+a}{a+s_0}$. Using this we obtain
\begin{align}
 h(s_0) &\geqslant \frac{(a+s_0)(1+a)}{a(1-s_0)+(1+a)s_0}=1+a
\end{align}
which completes the proof.
\qed\end{proof}
\noindent The lemma follows now immediately by combining \eqref{eq:predicting_probability_estimate} and the last claim.
\qed\end{proof}

\section{Proof of Lemma \ref{thm:main_function_properties}}

\begin{proof}[of Lemma]
It is easy to see that $\lim_{d\rightarrow 0^{+}} g(d) = \ell$. We have
\begin{align}\label{eq:main_function_derivative1}
 \frac{\partial g(d)}{\partial d} = \frac{(1-d)^{\ell-1}(d(\ell-1)+1)-1}{d^2}
\end{align}
Using the inequality $1-d \leqslant \mathrm{e}^{-d}$ we obtain
\begin{align*}
\frac{\partial g(d)}{\partial d} & \leqslant \frac{\mathrm{e}^{-d(\ell-1)}\left(d(\ell-1)+1 \right)-1}{d^2} \leqslant 0
\end{align*}
Where the second inequality follows from the inequality $\mathrm{e}^{s} \geqslant 1+s$ applied for $s=d(\ell-1)$. This proves (a). The second derivative is given by 
\begin{align}\label{eq:main_function_derivative2}
 \frac{\partial^2 g(d)}{\partial d^2} & = -\frac{(1-d)^{\ell-2}\left(2 + 2d(\ell-2) + d^2((\ell-2)^2+\ell-2) \right)-2}{d^3} 
\end{align}
Using $1-d\leqslant \mathrm{e}^{-d}$ and applying the inequality $\mathrm{e}^{s} \geqslant 1+s+\frac{1}{2}s^2$, which holds for $s\geqslant 0$, for $s = d(\ell-1)$ we obtain
\begin{align}\label{eq:main_function_derivative2_estimate}
 \frac{\partial^2 g(d)}{\partial d^2} & = -\frac{(1-d)^{\ell-2}\left(2 + 2d(\ell-2) + d^2((\ell-2)^2+\ell-2) \right)-2}{d^3} \nonumber \\
 & \geqslant  -\frac{(1-d)^{\ell-1}\left(2 + 2d(\ell-1) + d^2(\ell-1)^2 \right)-2}{d^3} \nonumber \\ 
 & \geqslant  -\frac{\mathrm{e}^{-d(\ell-1)}\left(2 + 2d(\ell-1) + d^2(\ell-1)^2 \right)-2}{d^3} \nonumber \\
 & \geqslant -\frac{2-2}{d^3} = 0,
\end{align}
which proves (b). Finally, note that by convexity we have
\begin{align}
 g(d_2)-g(d_1) \geqslant  (d_2-d_1) \cdot \left. \frac{\partial g(d)}{\partial d} \right|_{d=d_1}. 
\end{align}
Since $g(d) > 0$ and $\frac{\partial \ln g(d)}{\partial d} = \frac{\partial g(d)}{\partial d}/g(d)$ we can rewrite this as
\begin{align}
 \frac{g(d_2)-g(d_1)}{g(d_1)} \geqslant  (d_2-d_1) \cdot \left. \frac{\partial \ln g(d)}{\partial d}  \right|_{d=d_1}. 
\end{align}
Note that the function $d\rightarrow \ln g(d)$ is convex, as the composition of the convex function $g(\cdot)$ and the convex increasing function $\ln(\cdot)$. Therefore,
\begin{align}
 \frac{\partial \ln g(d)}{\partial d} \geqslant  \left.\frac{\partial \ln g(d)}{\partial d} \right|_{d=0} =-\frac{\ell-1}{2}
\end{align}
Combining the last two inequalities yields
\begin{align}
 \frac{g(d_2)-g(d_1)}{g(d_1)} > -\frac{\ell}{2}\cdot (d_2-d_1),\quad d_2-d_1 >0.
\end{align}
which completes the proof of (c).
\qed\end{proof}